\documentclass[a4paper]{article}
\usepackage{titling}
\usepackage{a4wide}
\usepackage{amsmath}
\usepackage{amsfonts}
\usepackage{amssymb}
\usepackage{amsthm}
\usepackage{bbm}
\usepackage{hyperref}
\usepackage{cleveref}
\usepackage[utf8]{inputenc}
\usepackage{color}
\usepackage{url}
\usepackage{appendix}
\usepackage{csquotes}
\usepackage{mathtools}

\MakeOuterQuote{"}

\newtheorem{theorem}{Theorem}[section]
\newtheorem{corollary}[theorem]{Corollary}
\newtheorem{lemma}[theorem]{Lemma}

\newtheorem{example}[theorem]{Example}
\newtheorem{definition}[theorem]{Definition}
\numberwithin{equation}{subsection}
\numberwithin{theorem}{subsection}

\newcommand{\UHaddress}{\em University of Helsinki, Department of Mathematics 
and Statistics\\
\em P.O. Box 68, FI-00014 Helsingin yliopisto, Finland}
\newcommand{\email}[1]{E-mail: \tt #1}
\newcommand{\emailkalle}{\email{kalle.koskinen@helsinki.fi}}

\title{Infinite volume Gibbs states of the generalized mean-field orthoplicial model}
\author{Kalle Koskinen\thanks{\emailkalle} \\[1em] \UHaddress}
\date{\today}

\begin{document}
\begin{titlingpage}
\maketitle
\begin{abstract}
\noindent
The generalized mean-field orthoplicial model is a mean-field model on a space of continuous spins on $\mathbb{R}^n$ that are constrained to a scaled $(n-1)$-dimensional $\ell_1$-sphere, equivalently a scaled $(n-1)$-dimensional orthoplex, and interact through a general interaction function. The  finite volume Gibbs states of this model correspond to singular probability measures. In this paper, we use probabilistic methods to rigorously classify the infinite volume Gibbs states of this model, and we show that they are convex combinations of product states. The predominant methods utilize the theory of large deviations, relative entropy, and equivalence of ensembles, and the key technical tools utilize exact integral representations of certain partition functions and locally uniform estimates of expectations of certain local observables. 
\end{abstract}
\end{titlingpage}
\tableofcontents
\addtocontents{toc}{\protect\thispagestyle{empty}}
\newpage
\section{Introduction} \label{sec:Introduction}
\noindent
The purpose of this paper is to present rigorous probabilistic methods to compute and classify the large $n$-limits of integrals of the form
\begin{align} \label{def:primary integral}
\mu_n^g [f \circ \pi_I] := \frac{1}{Q_n (g)}\int_{\mathbb{R}^n} d \phi \ e^{n g \left( \frac{1}{n} \sum_{i=1}^n \phi_i \right)} \delta \left( \sum_{i=1}^n |\phi_i| - n \right) (f \circ \pi_I)(\phi) ,
\end{align}   
where $g : \mathbb{R} \to \mathbb{R}$ is a "sufficiently regular" function which will be referred to as an interaction function, $d \phi$ is the Lesbesgue measure on $\mathbb{R}^n$, $\delta (\cdot)$ is formally a delta function, $f \in C_b (\mathbb{R}^I)$, where $C_b (\mathbb{R}^I)$ is the space of continuous bounded functions on a finite index set $I \subset [n] := \{ 1,2,..., n\}$, $\pi_I  :\mathbb{R}^n \to \mathbb{R}^I$ is the canonical coordinate projection, and $Q_n (g)$ is a normalization constant, which will be referred to as the partition function, which make $\mu_n^g$ into a probability measure. The main result in this paper is given in \cref{thm:limiting state full}, and it constitutes a full characterization of the infinite volume Gibbs states corresponding to the models given by the probability measures in \cref{def:primary integral} given some regularity of the interaction function $g$.
\\
\\
We will refer to the probability measure $\mu_n^g$ as a finite volume Gibbs state and the infinite volume limit, i.e. the large $n$-limit, when it exists, will be referred to as the infinite volume Gibbs state. We refer to \cref{sec:Heuristics} for a complete definition and discussion of the notion of infinite volume Gibbs state in this context.
\\
\\
At a heuristic level, to make such finite volume Gibbs states rigorous, we use the fact that the constraint function inside the delta function, when restricted to an orthant of $\mathbb{R}^n$, precisely defines a uniform measure over a scaled $(n-1)$-dimensional simplex. This method is used in \cref{sec:Microcanonical probability measures}. Since the $(n-1)$-dimensional $\ell_1$-sphere corresponds to the $(n-1)$-dimensional orthoplex, we refer to this model as the generalized mean-field orthoplicial model. This naming convention is similar to the convention used for the mean-field spherical model, see \cite{Kastner2006}, but the constraint changes from the $\ell_2$-sphere to the orthoplex.
\\
\\
Mean-field models of equilibrium statistical mechanics have been studied extensively as toy models of spins on various types of spaces, and the most famous model belonging to this class is the Curie-Weiss model, see \cite{Ellis1978,Ellis2006, Friedli2017} . The classical spin-$\frac{1}{2}$-Curie-Weiss model is a relatively simple exactly solvable model with interesting statistical mechanical phenomena such as phase transitions, anomalous scaling, infinite volume Gibbs states, etc. In addition, the model can be generalized in a variety of ways while retaining the essential simplicity of the models. Such generalizations are for instance the modifications of the interaction function, see \cite{Eisele1988}, additions of external random fields, see \cite{Matos1991}, and modifications to the ambient space, see \cite{Ellis1978}. In this direction, there are also modifications to the entire measure of the ambient spin space, where one changes the product structure to a constrained singular measure such as the uniform measure on a scaled sphere, see \cite{Kastner2006}. The model here also falls into this category, since the ambient space of spins is assumed to be constrained to scaled spheres in the $\ell_1$ norm.
\\
\\
To our knowledge, the orthoplicial model and the methods presented to solve the various problems associated with the orthoplicial model are novel in the literature, and known methods, such as in \cite{Ellis1978}, are not necessarily applicable. A particular approach, which is quite natural, is to try to swap the delta function for an appropriately parametrized exponential function, and solve rigorously the problem of interchanging such functions. This approach, however, fails, see \cref{sec:Heuristics}. Instead, the general method introduced in this paper, presented heuristically in \cref{sec:Heuristics}, is to go from a singly delta constrained probability measure to one which is doubly constrained. This particular doubly constrained probability measures is more tractable, and one can characterize both its limiting probability measure, and the uniform parametric convergence it has in the large $n$-limit. The limiting measure of the doubly constrained measure has a product structure, and, subject to further analysis, we find the general result that for a wide variety of suitable interaction functions, the limiting states of the model are convex combinations of product states. The suitable swapping between constrained and non-constrained measures is one aspect of the equivalence of ensembles, see \cite{Touchette2015}.
\\
\\
Let us now remark on the main methods and concepts used in this paper in more detail. The first step involves writing the finite volume Gibbs state as an integral mixture of probability measures such that the mixture acts on two variables which parametrize a doubly constrained measure which we will call a microcanonical probability measure. This step is carried out in \cref{sec:Main results} and \cref{sec:Microcanonical probability measures}. The general strategy is then to utilize a type of generalized dominated convergence theorem where both the integrating measure and the functions we are integrating are varying, see \cref{thm:general convergence}.
\\
\\
Using relative entropy methods, we are able to show that the difference in expectations of local observables of microcanonical probability measure and the completely unconstrained probability measures, referred to as the grand-canonical probability measure, depend explicitly on their corresponding statistical mechanical entropies, see \cref{thm:fundamental inequality}. Using large deviations methods, we are able to prove a broadly applicable theorem which allows one to prove locally uniform convergence of the finite volume microcanonical entropies using the concavity of the microcanonical entropies along with convergence of the grand-canonical entropy, see \cref{thm:entropy convergence}. Note that we are using a non-standard terminology by referring to the normalized logarithm of a partition function as an entropy irrespective of the ensemble the partition function comes from. We should emphasize that \cref{thm:entropy convergence} formalizes, at least at this level of regularity, the notion that one can rigorously deduce many properties of the limiting microcanonical entropy from the grand canonical entropy, which is typically far more tractable mathematically. 
\\
\\
For the orthoplicial model, one can easily verify some of the conditions of \cref{thm:entropy convergence} that are required, since the corresponding grand-canonical probability measure is a product state. Of some methodological interest is the fact that we use the notion of Lorentzian polynomials, see \cite{Braenden2020}, to prove that the microcanonical entropies are log-concave functions. The end result is that by combining together, \cref{thm:fundamental inequality}, \cref{thm:free energy properties}, and \cref{thm:limiting microcanonical entropy}, we obtain the locally uniform convergence of the difference of expectations of local observables for the microcanonical and grand-canonical probability measure in the form of \cref{thm:locally uniform convergence expectations}.
\\
\\
As for the mixture probability measure, we begin by once again applying the general theorem \cref{thm:entropy convergence}, to deduce the entropy of the corresponding canonical model, see \cref{thm:half-constrained ensemble free energy}, which is directly related to our model with a linear interaction function. Using tilting, we are then able to show that the mixture probability measures satisfies a large deviations principle, see \cref{thm:full free energy}. Using large deviations techniques found in \cref{sec:Large deviations and weak convergence}, we are able to already classify the limiting states of our model for a variety of relevant non-trivial interaction functions, see \cref{ex:CW non-vanishing} for the quadratic mean-field interaction with a non-vanishing magnetic field, see \cref{ex:CW vanishing}, for the quadratic mean-field interaction without an external magnetic field, and see \cref{thm:single point} and \cref{thm:double point} for the rigorous results concerning these two examples. 
\\
\\
In order to fully classify the limiting states of the model for more general interaction functions, we need an additional result concerning the microcanonial partition function which comes in the form of an exact generating function representation, see \cref{thm:microcanonical partition function integral}. The generating function that we obtain is a modified Bessel function of the first kind, and we utilize a particular integral representation of it. This allows one to fully characterize the weak convergence of the mixture probability measure by relating it to to Laplace-type integrals in three variables for which we can exactly deduce their asymptotics, see \cref{thm:exact asymptotics}. The exact result employs the notions of type, and maximal type, given in \cite{Ellis1978}, adapted to this particular model. The primary pair of results concerning this final result are \cref{thm:mixture weak full } and \cref{thm:limiting state full}, which can be summarized by stating that given sufficient regularity of the interaction function $g$, which are intimately related to properties of the limiting entropy, one is able to show that the limiting states are convex combinations of products states. 
\\
\\
In the literature, the closest works are \cite{Kastner2006} and \cite{Koskinen2020}, in which similar results, with entirely different methods, are produced for the so-called mean-field spherical model. Another similar work which considers a Berlin-Kac-type, see \cite{Berlin1952}, model with a spherical constraint is given in \cite{Lukkarinen2019}. From the pure mathematical perspective, non-interacting continuous models with multiple constraints have been considered in \cite{Chatterjee2017} and \cite{Nam2020}. These works both consider the particular phenomenon of condensation, and their approach could be described as probabilistic ones. For discrete two-constraint models, and formalism for the equivalence of ensembles for such models, see \cite{Grosskinsky2008}. In terms of methods, in \cite{Caputo2003}, there is an approach to proving a type of uniform convergence between constrained and non-constrained probability measures by adapting a uniform local central limit theorem. For a random-field model constrained to the sphere, a similar uniform convergence result between constrained and non-constrained probability measures is obtained in \cite{Koskinen2023}. Finally, we should also remark that this paper does not make use of the method of steepest descent, see \cite{Berlin1952}, nor do we rely on characteristic functions in any particular way to complete any of the proofs.
\subsection{Reading guide}
\noindent
This paper is primarily organized so that a majority of the concepts and methods without proofs can be gathered by reading the introduction contained in \cref{sec:Introduction} and the heuristics contained in \cref{sec:Heuristics}. These sections do not contain any proofs, but they do contain some definitions and outline the basic approach to the problems in this paper. 
\\
\\
The statements of the results, some important intermediate results, short or simple proofs, and relevant expository computations are done in \cref{sec:Main results}. The more involved proofs or methods are contained in \cref{sec:Intermediate results and proofs}. Note that \cref{sec:Intermediate results and proofs} also contains an entire subsection devoted to some results in theory of large deviations, see \cref{sec:Large deviations and weak convergence}, and the basic concepts and properties of relative entropy are given in \cref{sec:Relative entropy and local observables}.
\section{Heuristics} \label{sec:Heuristics}
\noindent
The functions $f$ used in \cref{def:primary integral} will be referred to as local functions and their associated finite index sets $I$ will be referred to as local index sets. Such local functions $f$ are naturally functions on $\mathbb{R}^n$ for large enough $n$ by using the coordinate projection $\pi_I : \mathbb{R}^n \to \mathbb{R}^I$, and representing them as a composition $f \circ \pi_I$. If one is able to resolve the large-$n$ limits of integrals of the form given in \cref{def:primary integral}, then one is able to specify, in the limit, the "expectations" of a large class of local observables. In doing so, subject to other regularity conditions on this limiting state one is able to produce a genuine probability measure on $\mathbb{R}^\mathbb{N}$. From now on, we will omit the coordinate projection $\pi_I$, and simple write the expectation with respect to a local function $f$ without the composition, unless it becomes pertinent for a specified reason. We will use the following definition of weak convergence and limit points of probability measures. 
\begin{definition} A sequence of probability measures $\mathcal{G} := \{ \mu_n \}_{n \in \mathbb{N}}$, such that each $\mu_n$ is a probability measure on $\mathbb{R}^n$, is said to converge weakly to a probability measure $\mu_\infty$ on $\mathbb{R}^\mathbb{N}$ if 
\begin{align*}
\lim_{n \to \infty} \mu_n [f] = \mu_\infty [f]
\end{align*}
for any $f \in C_b (\mathbb{R}^I)$.
\\
\\
The set of limit points $\mathcal{G}_\infty$ of $\mathcal{G}$ is given by
\begin{align*}
\mathcal{G}_\infty := \left\{ \mu \in \mathcal{P} (\mathbb{R}^\mathbb{N}) : \exists \{ n_k \}_{k \in \mathbb{N}}, \ \lim_{k \to \infty} \mu_{n_k} = \mu  \right\} ,
\end{align*}
where the limit is understood in the sense of the weak limit given here.
\end{definition}
\noindent 
There are simple extensions, see \cite{Koskinen2023} for an extension by "tensoring on 0" to the remaining $\mathbb{N} \setminus [n]$ components, that make the probability measure $\mu_n$ in this definition into probability measures on $\mathbb{R}^\mathbb{N}$, and using these extensions the definitions above are equivalent to the standard definitions of weak convergence of probability measure on Polish spaces, and the notion of limit points is to be understood as limit points with respect to the Lévy–Prokhorov metric. For our purposes, understanding that we are predominantly interested in studying the limit of expectations of local observables is sufficient for the contents of this paper. 
\\
\\
Using this notation, we are then interested in studying and classifying the structure and content of the sets $\mathcal{G}^g$, corresponding to the sequence of probability measures $\{ \mu_n^g \}_{n \in \mathbb{N}}$ specified in their functional form in \cref{def:primary integral}, which will be called the collection of finite volume Gibbs states, and $\mathcal{G}_\infty^g$, which will be called the collection of infinite volume Gibbs states, and their dependence on the interaction function $g$. 
\\
\\
The prototypical interaction function $g$ of this paper is based on the the Curie-Weiss Hamiltonian $H^J_{\operatorname{CW}, n} : \mathbb{R}^n \to \mathbb{R}$ given by
\begin{align*}
H^J_{\operatorname{CW}, n}(\phi) := - \frac{J}{2 n} \sum_{i,j = 1}^n \phi_i \phi_j = n \left( - \frac{J}{2} \left( \frac{1}{n} \sum_{i=1}^n \phi_i  \right)^2 \right) ,
\end{align*}
where $J > 0$ is a coupling constant, with the associated interaction function $g^{\beta, J} : \mathbb{R} \to \mathbb{R}$ given by
\begin{align*}
g^{\beta, J} (m) := \frac{\beta J}{2} m^2 ,
\end{align*}
where $\beta > 0$. With this interaction function, the probability measure in \cref{def:primary integral} takes the form
\begin{align*}
\mu^{\beta,J}_n [f] := \frac{1}{Q_n (\beta,J)}\int_{\mathbb{R}^n} d \phi \ e^{\frac{\beta J}{2 n} \sum_{i,j=1}^n \phi_i \phi_j} \delta \left( \sum_{i=1}^n |\phi_i| - n \right) f (\phi), 
\end{align*} 
and can be seen to contain two competing weights in the integrand: the interaction function gives larger weight to fields $\phi$ in which the components are of the same sign and as large as possible, this type of behaviour is why we refer to this interaction as ferromagnetic, while the delta function terms constrains the size aspect of the interaction. It is this competition which produces the non-trivial nature of the limiting state.
\\
\\
From this recipe of going from the Hamiltonian to the interaction function $g$, we can produce a number of "generalized" interactions such as $k$-body interactions corresponding to interaction functions of polynomial-type
\begin{align*}
g (m) := \sum_{j=1}^k \alpha_j m^{2j} ,
\end{align*}
where $\alpha_j$ are some real constants, even convex smooth interactions intended to model non-polynomial ferromagnetic interaction, and countless others which might be of interest.
\\
\\
The problem described here is well understood for models where the delta function is replaced by a product of density function, see \cite{Ellis1978}. Let us now remark on the connection between these types of generalized Curie-Weiss models, and the generalized mean-field orthoplicial model.
\\
\\
Formally, using delta functions, we have
\begin{align} \label{eq:formal mixture calculation}
&\int_{\mathbb{R}^n} d \phi \ e^{n g \left( \frac{1}{n} \sum_{i=1}^n \phi_i \right)} \delta \left( \sum_{i=1}^n |\phi_i| - n \right) f (\phi) \\ &= n \int_{-1}^1 dm \ e^{n g (m)} \int_{\mathbb{R}^n} d \phi \ \delta \left( \sum_{i=1}^n \phi_i - m n\right)\delta \left( \sum_{i=1}^n |\phi_i| - n \right)  f (\phi) \notag \\
&=  n \int_{-1}^1 dm \ e^{n g (m)} Z_n (m n,n) \nu_n (m,1) [f] , \notag
\end{align}
where
\begin{align} \label{def:microcanonical probability measure}
\nu_n (m,\rho) [f] := \frac{1}{Z_n (m n, \rho n)}  \int_{\mathbb{R}^n} d \phi \ \delta \left( \sum_{i=1}^n \phi_i - m n\right)\delta \left( \sum_{i=1}^n |\phi_i| - \rho n \right)  f (\phi) ,
\end{align}
where $\rho > 0$, $|m| \leq \rho$, and $Z_n (m n, \rho n)$ is a normalization constant which makes $\nu_n (m, \rho)$ into a probability measure. The values of $(m, \rho)$ for which the probability measure $\nu_n (m, \rho)$ exists in some formal sense are given by pairs satisfying $\rho > 0$, and $|m| \leq \rho$. These statements can be heuristically guessed "geometrically" by considering the intersection of hyperplanes with the $\ell_1$-spheres. For reasons which will become clear later, we will consider the interior of this set of existence, given and denoted by $\mathcal{A} := \{ (m, \rho) : \rho > 0, |m| < \rho\}$. Returning to \cref{def:primary integral}, we see that
\begin{align} \label{def:primary integral mixture form}
\mu_n^g [f] = \frac{n}{Q_n (g)} \int_{-1}^1 dm \ e^{n g (m)} Z_n (m,1) \nu_n (m,1) [f] .
\end{align}
In this form, the finite volume Gibbs state is written as an integral mixture of another probability measure. 
\\
\\
Although the original problem constrained the integrals to the $\ell_1$ ball of radius $n$, we have suggestively modified the notation so as to include the other possible values of the radius. This suggestive notation is due to the principle or phenomenon of the equivalence of ensembles, see \cite{Touchette2015}. We will refer to the probability measure $\nu_n (m, \rho)$ given formally in \cref{def:microcanonical probability measure}  as the  microcanonical probability measure. This probability measure is constrained by two functions $M_n, N_n : \mathbb{R}^n \to \mathbb{R}$ given by
\begin{align} \label{def:magnetization and particle number}
M_n (\phi) := \sum_{i=1}^n \phi_i, \ N_n (\phi) := \sum_{i=1}^n |\phi_i| .
\end{align}
We will refer to these functions as macrostates and the individual functions will be referred to as the magnetization and particle number respectively. In this paper, we will often refer to either ensembles or probability measures when discussing a particular thermodynamic model. Integrals with delta functions of the macrostates are referred to as constrained, and whenever we replace a delta function by some non-singular "function" of a macrostate, we are moving toward a less constrained state. With this perspective in mind, we will focus on the connection between the microcanonical probability measure and the grand canonical probability measure $\eta (\beta, \mu)$ on $\mathbb{R}^\mathbb{N}$ given by its action on $f \in C_b (\mathbb{R}^I)$ given by
\begin{align} \label{def:grand canonical probability measure}
\eta (\beta, \mu) [f \circ \pi_I] := \frac{1}{q(\beta, \mu)^{|I|}} \int_{\mathbb{R}^I} d \phi \ e^{- \beta \sum_{i \in I} \phi_i - \mu \sum_{i \in I}|\phi_i|} f (\phi) ,
\end{align}
where $\mu > 0$, $|\beta| < \mu$, and $q(\beta, \mu)^{|I|}$ is  a normalization constant making the finite marginals into probability measures. One can compute, by direct integration, that
\begin{align}
q(\beta, \mu) := \frac{1}{\mu + \beta} + \frac{1}{\mu - \beta} .
\end{align}
Note that, strictly speaking, the grand canonical probability measure should refer to the probability measure obtained from $\eta (\beta, \mu)$ by considering its marginal distribution on the index set $[n]$. 
\\
\\
The equivalence of ensembles principle states that, subject to some yet to be verified properties of the microcanonical and grand canonical partition functions, there are a number of ways in which these two probability measures are the same. For our purposes, we will utilize ideas stemming from the ensemble equivalence principle corresponding in some sense to thermodynamic, macrostate, and measure level equivalence of these probability measures. For a complete view on the principle of the equivalence of ensembles, see \cite{Touchette2015}.
\\
\\
To that end, we will need the finite and infinite volume specific microcanonical entropies $s_n, s : \mathcal{A} \to \mathbb{R}$ given respectively by
\begin{align} \label{def: microcanonical entropy}
s_n (m, \rho) := \frac{1}{n} \ln Z_n (m n, \rho n), \ s (m, \rho) := \lim_{n \to \infty} s_n (m, \rho) .
\end{align}
In addition, for the grand canonical ensemble we will need the finite and infinite volume specific entropies $f_n,f : \mathcal{A} \to \mathbb{R}$ given respectively by
\begin{align} \label{def: grand canonical entropy}
f_n (\beta, \mu) := \frac{1}{n} \ln q (\beta, \mu)^n, \ f (\beta, \mu) := \lim_{n \to \infty} f_n (\beta, \mu) .
\end{align}
Note the sign conventions used here. We will omit the specific part in their naming, and refer simply to entropies. For this particular model, as for all product state models, we trivially have $f_n (\beta, \mu) = f (\beta, \mu) = \ln q (\beta, \mu)$.
\\
\\
Using the entropies, we can rewrite \cref{def:primary integral mixture form} as
\begin{align}
\mu_n^g [f] = \frac{n}{Q_n (g)} \int_{-1}^1 dm \ e^{n (g (m) + s_n (m,1))} \nu_n (m,1) [f] .
\end{align}
The first type of equivalence property that we wish to utilize is the following pair of relations
\begin{align}
\sup_{(m, \rho) \in \mathcal{A}} \{ s (m, \rho) - \beta m - \mu \rho \} = f(\beta, \mu), \ \inf_{(\beta, \mu) \in \mathcal{A}} \{ f(\beta, \mu) + \beta m + \mu \rho \} = s (m, \rho) .
\end{align}
This relation is practically equivalent to that of two functions being Legendre conjugates, see \cite{Rockafellar1997}. Since we already have a closed form for $f(\beta, \mu)$, we may extract the form of $s(m, \rho)$ if this relation holds. 
\\
\\
The second equivalence property is the parameter matching scheme given by
\begin{align}
\eta (\beta, \mu) \left[ \frac{M_n}{n} \right] = m, \ \eta (\beta, \mu) \left[ \frac{N_n}{n} \right] = \rho . 
\end{align}
If for every pair $(m, \rho) \in \mathcal{A}$ there exists a corresponding pair $(\beta, \mu) \in \mathcal{A}$ satisfying the above relations and vice versa, then these corresponding pairs of values are the values for which we would expect the probability measures to be the same. We will use the notations $m (\beta, \mu)$, $\rho (\beta, \mu)$, $\beta(m, \rho)$, and $\mu (m, \rho)$ for this bijection. This bijection is intimately connected to the first equivalence property through the Legendre conjugates.
\\
\\
The final form of equivalence is then the rough statement that in the large $n$-limit, we have
\begin{align}
\nu_\infty (m, \rho) [f] := \lim_{n \to \infty} \nu_n (m, \rho) [f] = \eta(\beta (m, \rho), \mu (m, \rho)) [f] 
\end{align}
for local functions $f \in C_b (\mathbb{R}^{I})$.
\\
\\
If we now return to \cref{def:primary integral mixture form}, the heuristic behaviour of the model in the large $n$-limit is roughly speaking that
\begin{align}
\mu_n^g [f] \approx \left( \int_{-1}^1 dm \ e^{n (g(m) + s (m,1))} \right)
^{-1} \int_{-1}^1 dm \ e^{n (g(m) + s (m,1))} \nu_\infty (\beta (m,1), \mu (m,1)) [f] ,
\end{align}
and using the Laplace method, see \cite{Wong2001}, one would expect that
\begin{align}
&\left( \int_{-1}^1 dm \ e^{n (g(m) + s (m,1))} \right)
^{-1} \int_{-1}^1 dm \ e^{n (g(m) + s (m,1))} \nu_\infty (m,1) [f] &\approx \int_{M^* (\psi^g)} \alpha(dm) \ \nu_\infty (m,1) [f] ,
\end{align}
where $\alpha$ is a probability measure on $[-1,1]$ and $M^* (\psi^g) \subset (-1,1)$ is the set of global maximizing points of the mapping $[-1,1] \ni m \mapsto \psi^g (m):= g(m) + s(m,1)$. This is to be expected since integrands of the form above have an exponential rate concentration to the global maximum points of the given function.
\\
\\
The connection between this model and the generalized Curie-Weiss model is now evident. The limiting states of both models are given by mixtures of product states. However, for this model, one cannot realize these limiting states without the $\ell_1$ constraint. To see this, let us consider the following integral
\begin{align*}
W_n (\beta, \mu) := \int_{\mathbb{R}^n} d \phi \ e^{\frac{\beta J}{2 n}  \sum_{i,j=1}^n \phi_i \phi_j - \mu \sum_{i=1}^n |\phi_i|} ,
\end{align*}
where $\mu > 0$ and $\beta \leq 0$. This would be the less constrained grand canonical partition function to which one would hope that an equivalence principle holds. The partition function here is not finite if $\beta > 0$. For the allowed values of $\beta$, using the Fourier transform of the Gaussian, we have
\begin{align*}
W_n (\beta, \mu) &= \frac{1}{\sqrt{2 \pi}} \int_{- \infty}^\infty dz \ e^{- \frac{1}{2} z^2} \left( \int_{-\infty}^\infty d \phi \ e^{i \sqrt{\frac{(-\beta) J}{n}} z \phi - \mu |\phi|} \right)^n \\
&= \sqrt{\frac{n}{2 \pi}} \int_{- \infty}^\infty dz \ e^{- \frac{1}{2} nz^2} \left( \frac{2 \mu}{\mu^2 + (-\beta) J z^2} \right)^n \\
&= (2 \mu)^n \sqrt{\frac{n}{2 \pi}} \int_{- \infty}^\infty dz \ e^{- n \left( \frac{1}{2} z^2 + \ln (\mu^2 + (-\beta) J z^2)\right)} .
\end{align*}
Since the function $z \mapsto \frac{1}{2} z^2 + \ln (\mu^2 + (-\beta) J z^2)$ is trivially minimized when $z = 0$, by the Laplace method, it follows that
\begin{align*}
\lim_{n \to \infty} \frac{1}{n} \ln W_n (\beta, \mu) = \ln (2 \mu) - \mu^2 .
\end{align*} 
Now, if we include the mixture measure form of this integral, it follows that
\begin{align*}
&\frac{1}{W_n (\beta, \mu)} \int_{\mathbb{R}^n} d \phi \ e^{\frac{\beta J}{2 n}  \sum_{i=1}^n \phi_i \phi_j - \mu \sum_{i=1}^n |\phi_i|} f (\phi) \\ &= \left(\int_{- \infty}^\infty dz \ e^{- n \left( \frac{1}{2} z^2 + \ln (\mu^2 + (-\beta) J z^2)\right)}  \right)^{-1}\int_{- \infty}^\infty dz \ e^{- n \left( \frac{1}{2} z^2 + \ln (\mu^2 + (-\beta) J z^2)\right)} \eta (i \sqrt{(- \beta) J} z, \mu) [f] ,
\end{align*}
where $f \in C_b (\mathbb{R}^I)$ is a local function, from which we have
\begin{align*}
\lim_{n \to \infty} \frac{1}{W_n (\beta, \mu)} \int_{\mathbb{R}^n} d \phi \ e^{\frac{\beta J}{2 n}  \sum_{i=1}^n \phi_i \phi_j - \mu \sum_{i=1}^n |\phi_i|} f (\phi) = \eta (0, \mu) [f] .
\end{align*}
As can be seen, the limiting state is trivial in the sense that it is a pure state, i.e. not a convex combination of any other probability measures, and it does not depend on $\beta \leq 0$. It is this property why it is desirable to study the $\ell_1$ constrained model, since the replacement of the product measure, for this particular model, with a delta function reproduces the non-trivial limiting states.
\\
\\
The heuristic is then that the limiting states of the model are mixtures of product states of the form given in \cref{def:grand canonical probability measure}, where the mixture probability measure is determined by the properties of the interaction function $g$. This is precisely what we will prove rigorously.
\\
\\
Before presenting the main results and proofs, let us remark on the what exactly is not rigorous, incorrect, or too formal in the above exposition. The delta functions appearing in \cref{def:primary integral} and \cref{def:microcanonical probability measure} are completely formal objects, and we will rigorously define the microcanonical probability measure on which we can actually preform non-formal computations. In particular, the formal calculation presented in \cref{eq:formal mixture calculation} is strictly speaking incorrect. For this particular model, it is important to take into consideration the "boundary values" of the set $\mathcal{A}$. That is to say, the admissible pairs which satisfy $\rho > 0$ and $|m| = \rho$ produce partition functions which can not be neglected if one wants to verify the formal calculation in \cref{eq:formal mixture calculation}. In addition, the form of equivalence of ensembles we have specified here are vague and unverified. We will verify these forms of equivalence explicitly, and they will be presented as lemmas.
\section{Main results} \label{sec:Main results}
\noindent
In this section, we present the main results, short or simple proofs, and expository computations concerning the main results.
\subsection{Locally uniform convergence of observables and entropy of the microcanonical ensemble}
\noindent
We begin by rigorously defining the microcanonical probability measure $\nu_n (m, \rho)$ from \cref{def:microcanonical probability measure} for $(m, \rho) \in \mathcal{A}$, and the so-called "boundary values" corresponding to $\rho > 0$ and $|m| = \rho$. This is done by identifying the microcanonical probability measure as a convex combination of products of uniform measures on simplexes. The uniform measures on simplexes are rigorously definable via the so-called flag coordinates, and these uniform measures are computationally tractable. Due to the large number of properties that need to be shown for the microcanonical probability measures, we dedicate an entire section, see \cref{sec:Microcanonical probability measures}, to the rigorous definition and methods of use of this particular probability measure. The key definitions are for the microcanonical probaiblity measures $\nu_n (m, \rho)$ and the microcanonical partion functions $Z_n (M,N)$, now defined in \cref{def:microcanonical probability measure rigorous}.
\\
\\
In this work, we will often refer to Polish spaces and probability measures on them. Whenever we do so without an explicit reference to a $\sigma$-algebra, we implicitly mean with respect to the Borel $\sigma$-algebra associated with the topology of the Polish space. The basic principle by which we will identify the infinite volume Gibbs states is presented in the following lemma.
\begin{lemma} \label{thm:general convergence}Let $X$ be a Polish space. If $\{ \mu_n \}_{n \in \mathbb{N}}$ is a sequence of probability measures on $X$ converging weakly to a probability measure $\mu$ on $X$, $K \subset X$ is a compact continuity set of $\mu$ such that $\operatorname{supp} (\mu) \subset K$,
and $\{ f_n \}_{n \in \mathbb{N}}$ is a sequence of uniformly bounded functions on $X$ converging uniformly on $K$ to a function $f$, then it follows that
\begin{align*}
\lim_{n \to \infty} \int_{X} \mu_n (dx) \ f_n(x) = \int_{K} \mu (dx) \ f(x) .
\end{align*}
\end{lemma}
\noindent
The proof, see \cref{sec:Relative entropy and local observables}, is an application of conditioning to $K$ and applying various weak convergence properties. 
\\
\\
With reference to \cref{def:primary integral} and \cref{def:primary integral mixture form}, using the definition and methods of \cref{sec:Microcanonical probability measures}, we can write the finite volume Gibbs states in the following form
\begin{align*}
\mu_n^g = \int_{\mathbb{R}} \kappa_n^g (dm) \ \nu_n (m,1) ,
\end{align*} 
where $\kappa_n^g$ is a probability measures on $\mathbb{R}$ supported by $[-1,1]$ with action on $f \in C_b (\mathbb{R})$ given by
\begin{align} \label{def:mixture form rigorous}
\kappa_n^g [f] := \frac{1}{Q_n (g)} &\int_{-1}^1 dm \ n e^{n g (m)} Z_n (m n, n) f(m) \\ &+ n e^{n g (1)} Z_n ( n, n) f(1) + n e^{n g(-1)} Z_n (-n,n) f (-1) , \notag
\end{align}
where the partition function then takes the following form
\begin{align} \label{def:partition function rigorous}
Q_n (g) := \int_{-1}^1 dm \ n e^{n g (m)} Z_n (m n, n) + n e^{n g (1)} Z_n (n,n)  + n e^{n g(-1)} Z_n (-n,n)  .
\end{align}
In light of \cref{thm:general convergence}, we have two goals. The first goal is to show that the collection of mixture probability measures $\{ \kappa_n^g \}_{n \in \mathbb{N}}$ converges weakly to some limiting probability measure, and that there exists a compact continuity set of this limiting probability measures which contains the support of the limiting probability measure. The second goal is to show that for a fixed $f \in C_b (\mathbb{R}^I)$ the collection of functions $\{ \nu_n (m,1) [f] \}_{n \in \mathbb{N}}$ understood as a collection of functions on the variable $m \in [-1,1]$ is uniformly bounded, which is immediate by the boundedness of $f$, and uniformly convergent on the required compact continuity set. 
\\
\\
In the heuristic sketch in the introduction, we did not pay any particular attention to the modes of convergence of the limiting objects. For this particular model, we are able to locally uniform convergence by relating the rate of convergence of local functions to the rate and mode of convergence of the finite volume entropies. This connection is described in the following fundamental inequality. 
\begin{lemma} \label{thm:fundamental inequality} For any finite index set $I \subset [n]$ and any pairs of values $(m, \rho) \in \mathcal{A}$ and $(\beta, \mu) \in \mathcal{A}$, we have
\begin{align*}
&\sup_{f \in C_b (\mathbb{R}^I), \ || f ||_\infty \leq 1} \left| \nu_n (m, \rho) [f] - \eta (\beta, \mu) [f] \right| \\ &\leq \sqrt{\frac{|I| (n - 2)}{2(n - 2 - |I|)} \left( \beta m + \mu \rho + f(\beta, \mu) - \frac{n}{n - 2} s_n (m, \rho) \right)}.
\end{align*}
\end{lemma}
\noindent
The proof of this result is an application of Pinsker's inequality for relative entropy, followed by the subadditivity property of relative entropy coupled with the permutation invariance of the microcanonical probability measure. For this model, we can exactly compute the relative entropy of the $(n-2)$:th marginal of the microcanonical probability measure from which we obtain the entropy terms in the above inequality. For the full proof, see \cref{sec:Relative entropy and local observables}.
\\
\\
If we were only interested in showing that the microcanonical probability measure converges to the grand canonical probability measure, it can be accomplished by studying the pointwise convergence of the entropies. However, since we want to prove locally uniform convergence, we need some additional regularity. The additional regularity that we will prove is that the sequence of finite volume microcanonical entropies are pointwise uniformly bounded, and that the microcanonical partition functions are log-concave functions on $\mathcal{A}$. By a classical result in convex analysis, see \cite[section 10]{Rockafellar1997}, once the pointwise limit of the finite volume microcanonical entropies is deduced, the convergence is immediately elevated to locally uniform convergence.
\\
\\
In some models, the grand canonical entropy is more computationally tractable than the microcanonical entropy. This is the case here as well and we will prove a general result which utilizes the aforementioned regularity properties of the microcanonical partition functions coupled with some additional regularity properties of the grand canonical entropy to 
prove a result, which might also be of general interest in other models.
\begin{theorem} \label{thm:entropy convergence}
Let $\{ Z_n \}_{n \in \mathbb{N}}$ be a sequence of log-concave functions $Z_n : n \mathcal{C} \to (0, \infty)$, where $\mathcal{C} \subset \mathbb{R}^m$ is a non-empty open convex set and $n \mathcal{C} := \{n c : c \in \mathcal{C} \}$, such that 
\begin{align*}
\sup_{n \in \mathbb{N}} \left| \frac{1}{n} \ln Z_n (n x)\right| < \infty
\end{align*}
for any $x \in \mathcal{C}$, and there exists a non-empty open convex set $\mathcal{C}' \subset \mathbb{R}^m$ such that 
\begin{align*}
\int_{n \mathcal{C}} d X \ e^{- \left< t, X \right>} Z_n (X) < \infty
\end{align*}
for all $t \in \mathcal{C}'$ and all $n \in \mathbb{N}$, where $\left< \cdot, \cdot\right>$ is the Euclidean inner product.
\\
\\
If the function $f : \mathbb{R}^m \to \mathbb{R} \cup \{ \pm \infty \}$ given by the mapping
\begin{align*}
f(t) := \lim_{n \to \infty} \frac{1}{n} \ln \int_{n \mathcal{C}} dX \ e^{- \left< t, X\right>} Z_n (X)
\end{align*}
exists and is a proper convex lower semi-continuous function of Legendre type which satisfies $\nabla [- f] \mathcal{C}' = \mathcal{C}$ then it follows that
\begin{align*}
\lim_{n \to \infty} \sup_{x \in K} \left| \frac{1}{n} \ln Z_n (nx) - \inf_{t \in \mathbb{R}^m} \{ \left< t, x \right> + f(t) \}  \right| = 0 ,
\end{align*}
for any compact set $K \subset \mathcal{C}$.
\end{theorem}
\noindent
The proof of this result, see \cref{sec:Large deviations and weak convergence}, requires definitions and notions from large deviations theory. We have dedicated an entire section, see \cref{sec:Large deviations and weak convergence}, to the relevant definitions, and results which can be deduced after establishing a large deviations principle. The proof itself uses a relative compactness argument concerning locally uniformly convergent subsequences, and a characterization of the limits of said subsequences using a large deviations principle. 
\\
\\
To apply this method to this model, we being by providing the sufficient regularity of the finite volume microcanonical entropies. 
\begin{lemma} \label{thm:microcanonical entropy regularity} The collection of finite volume microcanonical entropies $\{ s_n \}_{n \in \mathbb{N}}$ is pointwise uniformly bounded and concave on $\mathcal{A}$.
\end{lemma}
\noindent
The proof, see \cref{sec:Infinite volume entropies and states}, of log-concavity proceeds by identifying the microcanonical partition functions $Z_n$ as a composition of a bivariate Lorentzian polynomial of degree $n-2$ and a linear map. To prove the uniform pointwise boundedness, we use the positivity of the relative entropy between the $(n-2)$:th marginal of the microcanonical probability measure and the grand-canonical probability measure.
\\
\\
In light of \cref{thm:entropy convergence}, it remains to consider the mapping $f : \mathbb{R}^2 \to \mathbb{R}$ given by
\begin{align*}
f(\beta, \mu) := \lim_{n \to \infty} \frac{1}{n} \ln \int_{ \mathcal{A}} d M d N \ e^{- \beta M - \mu N} Z_n (M, N) ,
\end{align*}
where, in accordance with \cref{def:microcanonical partition function}, we have
\begin{align*}
Z_n (M,N) = \frac{1}{2} \sum_{k=1}^{n - 1} {n \choose k} \frac{\left( \frac{N + M}{2}\right)^{k - 1}}{(k - 1)!} \frac{\left( \frac{N - M}{2}\right)^{n - k - 1}}{(n - k - 1)!}
\end{align*}
for $(M,N) \in \mathcal{A}$.
\\
\\
It is immediate that if $(\beta, \mu) \not \in \mathcal{A}$, then $f(\beta,\mu) = \infty$. As for $(\beta, \mu) \in \mathcal{A}$, we can directly compute that
\begin{align*}
\int_{\mathcal{A}} d M d N \ e^{- \beta M - \mu N} Z_n (M, N) &=  \int_0^\infty d X \int_0^\infty d Y \ e^{- (\mu + \beta) X - (\mu - \beta) Y} \sum_{k=1}^{n-1} {n \choose k} \frac{X^{k - 1}}{(k-1)!} \frac{Y^{n - k - 1}}{(n - k - 1)!} \\
&=   \sum_{k=1}^{n-1} {n \choose k} \left( \frac{1}{\mu + \beta} \right)^k \left( \frac{1}{\mu - \beta} \right)^{n-k} \\
&=   \left( \frac{1}{\mu + \beta} + \frac{1}{\mu - \beta} \right)^n -  \left( \frac{1}{\mu + \beta} \right)^n -  \left( \frac{1}{\mu - \beta} \right)^n   .
\end{align*}
Computing the limit, it follows that
\begin{align*}
f(\beta, \mu) = \ln \left( \frac{1}{\mu + \beta} + \frac{1}{\mu - \beta} \right) = \ln q (\beta, \mu) .
\end{align*}
In summary, we have
\begin{align*}
f(\beta,\mu) = \begin{cases} \ln \left( \frac{1}{\mu + \beta} + \frac{1}{\mu - \beta} \right), &\ (\beta, \mu) \in \mathcal{A} \\ \infty, &\ (\beta, \mu) \not \in \mathcal{A} \end{cases} .
\end{align*}
We have included this calculation here to emphasize the fact that this calculation is relatively straightforward.
\\
\\
We present the relevant regularity conditions of the map $f : \mathbb{R}^2 \to \mathbb{R}$ in the following result.
\begin{lemma} \label{thm:free energy properties} The mapping $f : \mathbb{R}^2 \to \mathbb{R}$ is a proper convex lower semi-continuous function of Legendre type.
\\
\\
In addition, it follows that $(- \nabla [f]) \mathcal{A} = \mathcal{A}$, and 
\begin{align*}
\inf_{(\beta, \mu) \in \mathbb{R}^2} \{ \beta m + \mu \rho + f (\beta, \mu) \} &= \beta (m, \rho) m + \mu (m, \rho) \rho + f (\beta (m, \rho), \mu (m, \rho)) \\
&= 1 + \ln \left( \left( \sqrt{\frac{\rho + m}{2}} + \sqrt{\frac{\rho - m}{2}}\right)^2 \right) ,
\end{align*}
where $(\beta, \mu) := (-\nabla[f])^{-1} : \mathcal{A} \to \mathcal{A}$ is given by
\begin{align*}
\beta (m, \rho) := - \frac{\rho}{m} \frac{1}{\sqrt{\rho^2 - m^2}} + \frac{1}{m}, \ \mu (m, \rho) := \frac{1}{\sqrt{\rho^2 - m^2}} .
\end{align*}
\end{lemma}
\noindent
For the proof, which is completely computational, see \cref{sec:Infinite volume entropies and states}.
\\
\\
Combining together the regularity of the finite volume entropies from \cref{thm:microcanonical entropy regularity}, and the computations and verifications concerning the function $f$ given in \cref{thm:free energy properties}, we have the following result.
\begin{lemma} \label{thm:limiting microcanonical entropy} It follows that
\begin{align*}
\lim_{n \to \infty} \sup_{(m, \rho) \in K \subset \mathcal{A}} | s_n (m, \rho) - s (m, \rho)| = 0,
\end{align*}
for any compact set $K \subset \mathcal{A}$, where 
\begin{align*}
s (m, \rho) := 1 + \ln \left( \left( \sqrt{\frac{\rho + m}{2}} + \sqrt{\frac{\rho - m}{2}}\right)^2 \right)  ,
\end{align*}
for any $(m, \rho) \in \mathcal{A}$.
\end{lemma}
\noindent
Combining together \cref{thm:fundamental inequality}, \cref{thm:free energy properties}, and \cref{thm:limiting microcanonical entropy}, we have the following result concerning the mode of convergence of local observables of the microcanonical probability measures.
\begin{corollary} \label{thm:locally uniform convergence expectations}For any finite index set $I \subset [n]$, it follows that
\begin{align*}
\lim_{n \to \infty} \sup_{(m, \rho) \in K \subset \mathcal{A}}\sup_{f \in C_b (\mathbb{R}^I), \ || f ||_\infty \leq 1} \left| \nu_n (m, \rho) [f] - \eta (\beta (m, \rho), \mu (m, \rho)) [f] \right| = 0 .
\end{align*}
\end{corollary}
\noindent
Having established the compact-open convergence of the microcanonical probability measures, we move on to the weak convergence of the mixture probability measures.
\subsection{Limiting entropy and convergence of mixture probability measures}
\noindent
By the heuristics given, it is evident that the mixture probability measures $\{ \kappa_n (g) \}_{n \in \mathbb{N}}$ should converge, at an exponential rate, to the global maximizing points of some tilting function. This idea can be realized by proving that the mixture probability measures satisfy a large deviations principle. Since the full models have a general interaction function $g$, we will first prove a large deviations principle for linear $g$, and then use tilting to obtain the full large deviations principle. The following result considers the large deviations principle for a linear $g$.
\begin{lemma} \label{thm:half-constrained ensemble free energy} Let $\beta \in \mathbb{R}$, $g^\beta (m) := - \beta m$, $Q_n (\beta) := Q_n (g^\beta)$, and $\kappa_n^\beta := \kappa_n^{g^\beta}$.
\\
\\
Then, it follows that
\begin{align*}
\lim_{n \to \infty} \frac{1}{n} \ln Q_n (\beta) = \sup_{m \in [-1,1]} \{ s (m,1) - \beta m \} .
\end{align*}
Moreover, $\{ \kappa_n^\beta \}_{n=1}^\infty$ satisfies a large deviations principle with rate function $I^\beta : \mathbb{R} \to [0, \infty]$ given by
\begin{align*}
[-1,1] \ni m \mapsto I^\beta (m) := \sup_{m \in [-1,1]} \{ s (m,1) - \beta m\} - (s(m,1) - \beta m),
\end{align*}
and $I^\beta (m) = \infty$ for $m \not \in [-1,1]$.
\end{lemma}
\noindent
The proof, see \cref{sec:Infinite volume entropies and states}, follows the same strategy as the microcanonical entropy. Here, the log-concavity is proved by an application of the Prekopa-Leindler theorem, and pointwise uniform boundedness is a direct calculation. 
\\
\\
Since the previous result yields a large deviations principle for the mixture probability measures $\{ \kappa_n^0 \}_{n \in \mathbb{N}}$, corresponding to the choice of $g$ being identically $0$, as direct corollary of tilting, see \cite{Hollander2008}, we have the following large deviations principle for the full mixture probability measures.
\begin{corollary} \label{thm:full free energy} For any $g \in C_b ([-1,1])$, it follows that
\begin{align*}
\lim_{n \to \infty} \frac{1}{n} \ln Q_n (g) = \sup_{m \in [-1,1]} \{ g(m) + s (m,1) \} .
\end{align*}
Moreover, $\{ \kappa_n^g \}_{n=1}^\infty$ satisfies a large deviations principle with rate function $I^g : \mathbb{R} \to [0, \infty]$ given by
\begin{align*}
[-1,1] \ni m \mapsto I^g (m) := \sup_{m \in [-1,1]} \{ g(m) + s (m,1)\} - (g(m) + s (m,1)) ,
\end{align*}
and $I^g (m) = \infty$ for $m \not \in [-1,1]$.
\end{corollary}
\noindent
Whenever a sequence of probability measures satisfies a large deviations principle with some rate function, it is accompanied by a measure concentration result to the kernel of the rate function, see \cref{sec:Large deviations and weak convergence}. In this vein, consider the function ${\psi^g} : [-1,1] \to \mathbb{R}$ given by
\begin{align*}
{\psi^g}(m) := g(m) + s (m,1) .
\end{align*} 
It is clear that if $I^g (m^*) = 0$, then the point $m^*$ corresponds to a global maximum point of ${\psi^g}$ by definition, and vice versa. Denote the set of global maximizing points of ${\psi^g}$ by $M^* (\psi^g)$, and, by the previous observation, we have $\left(I^g\right)^{-1} \{ 0 \} = M^* (\psi^g)$.
\\
\\
We may now begin the classification of the infinite volume Gibbs states. As a first partial result, by combining together \cref{thm:general convergence}, \cref{thm:locally uniform convergence expectations}, and \cref{thm:limiting points large deviations}, we have the following result.
\begin{lemma} \label{thm:partial limiting point}Let $g \in C_b ([-1,1])$, and suppose that $M^* (\psi^g) \subset (-1,1)$. Then, it follows that
\begin{align*}
\mathcal{G}^g_\infty \subset \left\{ \int_{-1}^1 \kappa (dm) \ \eta (\beta (m,1), \mu (\beta, 1)) :  \kappa \in \mathcal{M}_1 ([-1,1]), \ \operatorname{supp} (\kappa) \subset M^* (\psi^g) \right\} .
\end{align*}
\end{lemma}
\noindent
The proof, see \cref{sec:Infinite volume entropies and states}, is a direct combination of the given results.
\\
\\
As a corollary, if we can deduce that there is exactly one global maximizing point of ${\psi^g}$ contained in the interval $(-1,1)$, then there is a unique infinite volume Gibbs state. This follows since the Dirac measure on a single point is the only probability measure supported on a single point. 
\begin{theorem} \label{thm:single point}
Let $g \in C^1 ([-1,1])$, and suppose that ${\psi^g}$ has a unique global maximizing point $m^* \in (-1,1)$. 
\\
\\
Then, it follows that
\begin{align*}
\lim_{n \to \infty} \mu_n^g = \eta (\beta(m^*,1), \mu (m^*,1)) .
\end{align*}
\end{theorem}
\noindent
There are two prototypical functions $g$ that fall into this category. One we have already seen which is $g^\beta(m) = - \beta m$ for $\beta \in \mathbb{R}$. Since $s$ is strictly concave it is easy to check that there is a unique global maximizing point of $\psi (g^{\beta})$. The other example is related to the Curie-Weiss Hamiltonian with an external field.
\begin{example} \label{ex:CW non-vanishing}Consider $g (m) := \frac{\beta J}{2} m^2 + \beta h m$, where $J > 0$, $\beta > 0$, and $h \not = 0$. Let us first remark that ${\psi^g}$ must attain its maximum on $[-1,1]$. Suppose first that $h > 0$. For any point point $m^* < 0$ of ${\psi^g}$, this point cannot be a global maximum point since ${\psi^g}(- m^*) > {\psi^g}(m^*)$. It follows that if there exists a global maximizing point, then it must be of the same sign as $h$. Let us continue now with the case where $h > 0$, and note that the other case is analogous. By direct computation, we have
\begin{align*}
\partial [\psi^g] (m) = 0 \iff \beta J m + \beta h - \frac{m}{\sqrt{1 - m^2}} \frac{1}{1 + \sqrt{1 - m^2}} = 0 .
\end{align*}
One can further compute that
\begin{align*}
\partial^3 [\psi^g](m) = \frac{2 m^5 + 4 m^3 - 9 m \left( \sqrt{1 - m^2}  + 1 \right)}{\left( \sqrt{1 - m^2} + 1 \right)^3 (1 - m^2)^{\frac{5}{3}}} ,
\end{align*}
and
\begin{align*}
\frac{2 m^4 + 4 m^2}{9 \left( \sqrt{1 - m^2} + 1 \right)} < \frac{2 + 4 }{9} = \frac{6}{9} < 1 \implies 2 m^5 + 4 m^3 - 9 m \left( \sqrt{1 - m^2}  + 1 \right) < 0
\end{align*}
for $0 < m \leq 1$. It follows that $\partial [\psi^g](m) < 0$ on $(0,1]$ and $\partial [\psi^g]$ is thus strictly concave. In addition, we have $\partial [\psi^g](0) = \beta h > 0$, and $\lim_{m \to 1^-} \partial [\psi^g](m) = -\infty$. Using these properties, it follows that there must exist a unique point $m^* \in (0,1)$ such that $\partial [\psi^g](m^*) = 0$. In addition, by strict concavity of $\partial [\psi^g]$, it follows that ${\psi^g}$ is monotonically increasing on $(0, m^*)$ and monotonically decreasing  on $(m^*,1)$ which implies that this $m^*$ is the unique global maximum point and it is contained on $(0,1)$. A similar argument shows that if $h < 0$, then there is a unique global maximum point contained in $(-1,0)$. 
\end{example}
\noindent
For the second type of interaction, we consider even functions $g \in C_b ([-1,1])$ such that $g$ has precisely two global maximizing points $m^+ \in (0,1)$, and $m^- = - m^+ \in (-1,0)$. For such even functions, by spin-flip symmetry, or by changing variables $m \mapsto - m$, it follows that
\begin{align*}
\frac{\kappa_n^g (B(m^+, \delta))}{\kappa_n^g (B(m^+, \delta)) + \kappa_n^g (B(m^-, \delta))} = \frac{1}{2} = \frac{\kappa_n^g (B(m^-, \delta))}{\kappa_n^g (B(m^+, \delta)) + \kappa_n^g (B(m^-, \delta))},
\end{align*}
for small enough $\delta > 0$. In particular, by \cref{thm:weight split}, it follows that
\begin{align*}
\lim_{n \to \infty} \kappa_n^g = \frac{1}{2} \delta_{m^+} + \frac{1}{2} \delta_{m^-} 
\end{align*}
weakly. By combining this simple result with \cref{thm:general convergence}, and \cref{thm:locally uniform convergence expectations}, we have the following result.
\begin{theorem} \label{thm:double point}
Let $g \in C_b ([-1,1])$ be an even function such that $M^* (\psi^g) = \{ m^+, m^-\}$, where $m^+ > 0$, and $m^- = - m^+$.
\\
\\
Then, it follows that
\begin{align*}
\lim_{n \to \infty} \mu_n^g = \frac{1}{2} \eta (\beta(m^+,1), \mu (m^+,1)) + \frac{1}{2} \eta (\beta(m^-,1), \mu (m^-,1)) .
\end{align*} 
\end{theorem}
\noindent
The prototypical example here is the Curie-Weiss Hamiltonian without an external field.
\begin{example} \label{ex:CW vanishing}Consider $g(m) := \frac{\beta J}{2} m^2$ where $\beta > 0$, and $J > 0$. We have
\begin{align*}
\partial [\psi^g] (m) &= m \left( \beta J  -  \frac{1}{\sqrt{1 - m^2}} \frac{1}{1 + \sqrt{1 - m^2}} \right), \\  \partial^2 [\psi^g] (m) &= \beta J  -  \frac{1}{\sqrt{1 - m^2}} \frac{1}{1 + \sqrt{1 - m^2}} - \frac{m^2 (2 \sqrt{1 - m^2} + 1)}{(\sqrt{1 - m^2} + 1)^2 (1 - m^2)^{\frac{3}{2}}}. 
\end{align*}
From the form of the first derivative, we see that ${\psi^g}$ cannot obtain a maximum at either end of the interval $[-1,1]$ and must thus be attained at a critical point in the open interval $(-1,1)$. There are now two options for the critical point, the first is that $m = 0$, from which we have
\begin{align*}
\partial [\psi^g](0) = 0, \ \partial^2 [\psi^g](0) = \beta J - \frac{1}{2} . 
\end{align*}
Due to the sign of the second derivative, this fails to be even a local maximum when $\beta J > \frac{1}{2}$, and whatever other critical point must be the global maximizing point if we are in this parameter range. The other case is that
\begin{align*}
\beta J = \frac{1}{\sqrt{1 - {m^\pm}^2}} \frac{1}{1 + \sqrt{1 - {m^\pm}^2}} \iff \sqrt{1 - {m^\pm}^2} = \frac{1}{2} \left( \sqrt{\frac{4}{\beta J} + 1} - 1 \right)  ,
\end{align*}
when $\beta J > \frac{1}{2}$. For other values of $\beta J$, there is no solution to this equation and we must conclude that the other critical point corresponds to the global maximizing point. We can conclude that when $\beta J \in \mathbb{R}$, then $m^* = 0$ is always a critical point, but it cannot be even a local maximizing point when $\beta J > \frac{1}{2}$, hence in this regime we must conclude that the pair of solutions $m^\pm$ given above are the only viable critical points, but since they are the only critical points, and the function must attain its maximum at a critical point, we may conclude that $m^\pm$ also correspond to global maximizing points of the function. When $\beta J < \frac{1}{2}$, the $m^* = 0$ critical point is the only critical point, and we can again conclude that this must then be the global maximizing point. If $\beta J = \frac{1}{2}$, we can check that both $m^\pm = 0$, and thus we again have a single critical point which must be a global maximizing point.
\end{example}
\noindent
The interactions described here are ones which can be dealt with without any further study of the structure of the function ${\psi^g}$. When there are no symmetries or unique global maximum points, one has to resort to other methods to resolve the limits. We will now present such methods for dealing with sufficiently smooth interaction functions that have multiple global maximizing points.
\subsection{Exact integral representations of the weights and full classification of the infinite volume Gibbs state}
\noindent
We will need a preliminary result concerning the microcanonical partition function in order to have better control of the mixture probability measures. We have the following generating function based representation of the microcanonical partition function.
\begin{lemma} \label{thm:microcanonical partition function integral} Let $(m, \rho) \in \mathcal{A}$.
\\
\\
Then, it follows that
\begin{align*}
Z_n (m n, \rho n) = \frac{2^{2n - 1} n^{n - 2} n!}{(2n)! \sqrt{\rho^2 - m^2} n^2} {2n \choose 2}  \frac{e^{-(n-1)}}{\pi^2} \int_0^\pi d \theta_1 \int_0^\pi d \theta_2 \ \cos \theta_1 \cos \theta_2 e^{(n-1) s (m, \rho, \theta_1, \theta_2)} , 
\end{align*}
where $s :\mathcal{A} \times [0, 2 \pi) \times [0, 2 \pi)$ is given by
\begin{align*}
s (m, \rho, \theta_1, \theta_2) := 1 + \ln \left(  \left( \sqrt{\frac{\rho + m}{2}} \cos \theta_1 + \sqrt{\frac{\rho - m}{2}} \cos \theta_2 \right)^2 \right) .
\end{align*}
\end{lemma}
\noindent
The proof of this representation, see \cref{sec:Asymptotics of the weights}, follows by using the convolution structure of the microcanonical partition function and identifying the generating function to be the product of modified Bessel functions of the second kind. The proof is concluded by differentiation of these Bessel functions.
\\
\\
In the previous result, we introduced the overloaded $s$ function by adding an angular dependence. We will differentiate between these functions by always specifying, in one form or another, the number of arguments the function takes.
\\
\\ 
In the following, we will specialize to functions $g$ that are infinitely continuously differentiable, and obtain there finitely many global maximum points in the interval $(-1,1)$. In light of \cref{thm:weight split}, our goal is to study quantities of the form
\begin{align*}
\frac{\kappa_n^g (B(m^*, \delta))}{\sum_{m^* \in M^* (\psi^g)}\kappa_n^g (B(m^*, \delta))} = \frac{\int_{m^* - \delta}^{m^* + \delta} dm \ e^{n (g(m) + s_n (m,1))}}{ \sum_{m^* \in M^* (\psi^g)} \int_{m^* - \delta}^{m^* + \delta} dm \ e^{n (g(m) + s_n (m,1))}} .
\end{align*}
Using \cref{thm:microcanonical partition function integral}, it follows that
\begin{align} \label{eq:laplace type integral}
&\frac{(2n)! n^2 \pi^2}{2^{2n - 1} n^{n - 2} n! {2n \choose 2} e^{-(n-1)}} \int_{m^* - \delta}^{m + \delta} dm \ e^{n (g(m) + s_n (m,1))} \\ &= \int_{m^* - \delta}^{m + \delta} dm \int_{0}^\pi d \theta_1 \int_0^\pi d \theta_2 \ \frac{\cos \theta_1 \cos \theta_2 e^{g(m)}}{\sqrt{1 - m^2}} e^{(n-1) (g(m) + s (m,1, \theta_1, \theta_2))} \notag \\
&= \int_{m^* - \delta}^{m + \delta} dm \int_{0}^\pi d \theta_1 \int_0^\pi d \theta_2 \ \frac{\cos \theta_1 \cos \theta_2 e^{g(m)}}{\sqrt{1 - m^2}} e^{(n-1) ({\psi^g}(m, \theta_1, \theta_2))} , \notag
\end{align}
where we have introduced the overloaded function ${\psi^g} : (-1,1) \times [0, 2 \pi) \times [0, 2 \pi)$ given by
\begin{align*}
{\psi^g}(m,\theta_1, \theta_2) := g(m) +  1 + \ln \left( \left(\sqrt{\frac{1 + m}{2}} \cos \theta_1 + \sqrt{\frac{1 - m}{2}} \cos \theta_2 \right)^2 \right) .
\end{align*}
We see that the integral in \cref{eq:laplace type integral} takes the form of a Laplace-type integral in three variables, and we expect that the local structure around the global maximum points of the overloaded function ${\psi^g}$ determine the exponential asymptotics of such integrals precisely. 
\\
\\
To that end, we present the following result which contains the relevant information concerning the structure and local asymptotics of the overloaded ${\psi^g}$ function.
\begin{lemma} \label{thm:tilting function taylor polynomial} Suppose that ${\psi^g}$ has a local maximizing point $m^*$ contained in the interval $(m^* - \delta, m^* + \delta)$, and there exists $k \in \mathbb{N}$ such that $\partial^{2k} [\psi^g] (m^*) < 0$ and $\partial^j [\psi^g] (m^*) = 0$ for all $1 \leq j \leq 2k - 1$.
\\
\\
Then, it follows that 
\begin{align*}
{\psi^g} (m^* + m, \theta_1, \theta_2) &= {\psi^g} (m^*) + \frac{1}{2} \partial_2^2 [{\psi^g}] (m^*, 0, 0) \theta_1^2 + \frac{1}{2} \partial_{3}^2 [{\psi^g}] (m^*, 0, 0) \theta_2^2 + \frac{1}{(2k)!} \partial^{2k} [\psi^g] (m^*) m^{2k}\\
&+ \sum_{|\alpha| = 3, \ \alpha_1 \not \in \{ 2, 3 \}} R_\alpha (m, \theta_1, \theta_2) (m, \theta_1, \theta_2)^\alpha + R_{(2k + 1,0,0)} (m, \theta_1, \theta_2) m^{2k + 1} ,
\end{align*}
where
\begin{align*}
R_\alpha (m, \theta_1, \theta_2) = \frac{|\alpha|}{\alpha!} \int_0^1 dt \ (1 - t)^{|\alpha| - 1} \partial_\alpha [{\psi^g}] ((m^*, 0,0) + t (m, \theta_1, \theta_2)) .
\end{align*}
In addition, 
\begin{align*}
&\lim_{n \to \infty} n \left( {\psi^g} \left(m^* + \frac{m}{n^{\frac{1}{2k}}}, \frac{\theta_1}{n^\frac{1}{2}}, \frac{\theta_2}{n^{\frac{1}{2}}}\right) - {\psi^g} (m^*)\right) \\ & = \frac{1}{2} \partial_2^2 [{\psi^g}] (m^*, 0, 0) \theta_1^2 + \frac{1}{2} \partial_{3}^2 [{\psi^g}] (m^*, 0, 0) \theta_2^2 + \frac{1}{(2k)!} \partial^{2k}[\psi^g] (m^*) m^{2k} .
\end{align*}
\end{lemma}
\noindent
The proof of this result, see \cref{sec:Asymptotics of the weights}, follows by developing the Taylor polynomial of the overloaded ${\psi^g}$ function around the point $(m^*, 0, 0)$, and using the fact that odd derivatives of cosines vanish when evaluated at $0$. The second statement simply follows by taking the limit.
\\
\\
From the previous result, we see that it is pertinent to introduce the following classification, which is directly adapted from \cite{Ellis1978}, of the global maxima of ${\psi^g}$.
\begin{definition} A global maximum point $m^* \in (-1,1)$ of ${\psi^g}$ is said to be of type $k(m^*) \in \mathbb{N}$ if $\partial^{2k} [\psi^g](m^*) < 0$ and $\partial^j [\psi^g] (m^*) = 0$ for all $1 \leq j \leq 2k - 1$. 
\\
\\
For a finite collection of global maximum points $M^* (\psi^g) \subset (-1, 1)$ of ${\psi^g}$, the maximal type $k_\infty (\psi^g)$ is given by $k_\infty (\psi^g) = \max_{m^* \in M^* (\psi^g)} k (m^*)$. The collection of global maximum points of maximal type $M_\infty^* (\psi^g)$ is given by $M_\infty^* (\psi^g) := \{ m^* \in (-1,1) : k(m^*) = k_\infty (\psi^g)  \}$.
\end{definition}
\noindent
Combining together \cref{thm:microcanonical partition function integral}, \cref{thm:tilting function taylor polynomial}, and the form given in \cref{eq:laplace type integral}, we have the following asymptotic result.
\begin{lemma} \label{thm:exact asymptotics}Suppose that ${\psi^g}$ has a single unique maximizing point $m^* \in (m^* - \delta, m^* + \delta)$ of type $k \in \mathbb{N}$.
\\
\\
Then, it follows that
\begin{align*}
&\lim_{n \to \infty} \frac{n^{\frac{1}{2k} + 1}\int_{m^* - \delta}^{m^* + \delta} dm \ e^{n (g(m) + s_n (m,1))}}{e^{n {\psi^g} (m^*)}}\frac{(2n)! n^2 \pi^2}{2^{2n - 1} n^{n - 2} n! {2n \choose 2} e^{- (n-1)}} \\ &= \frac{e^{g(m^*)}}{e^{{\psi^g} (m^*)} \sqrt{1 - {m^*}^2}} \int_{\mathbb{R}^3} d \theta_1 d \theta_2 d m \ e^{\frac{1}{2} \partial_2^2 [{\psi^g}] (m^*, 0, 0) \theta_1^2 + \frac{1}{2} \partial_{3}^2 [{\psi^g}] (m^*, 0, 0) \theta_2^2 + \frac{1}{(2k)!} \partial^{2k}[\psi^g] (m^*) m^{2k}}  .
\end{align*}  
\end{lemma}
\noindent
The proof of this result, see \cref{sec:Asymptotics of the weights}, is a standard application of the multivariate Laplace method.
\\
\\
From the previous result, denote $W_n (g, m^*, \delta)$ to be the quantity given by
\begin{align*}
W_n^g (m^*, \delta) := \frac{n^{\frac{1}{2k} + 1}\int_{m^* - \delta}^{m^* + \delta} dm \ e^{n (g(m) + s_n (m,1))}}{e^{n {\psi^g} (m^*)}}\frac{(2n)! n^2 \pi^2}{2^{2n - 1} n^{n - 2} n! {2n \choose 2} e^{-(n-1)}} ,
\end{align*}
and its limit $W (g, m^*)$ given by
\begin{align*}
W^g (m^*) &:= \lim_{n \to \infty} W_n^g (m^*, \delta) \\
&=\frac{e^{g(m^*)}}{e^{{\psi^g} (m^*)} \sqrt{1 - {m^*}^2}} \int_{\mathbb{R}^3} d \theta_1 d \theta_2 d m \ e^{\frac{1}{2} \partial_2^2 [{\psi^g}] (m^*, 0, 0) \theta_1^2 + \frac{1}{2} \partial_{3}^2 [{\psi^g}] (m^*, 0, 0) \theta_2^2 + \frac{1}{(2k)!} \partial^{2k} [\psi^g] (m^*) m^{2k}} .
\end{align*}
To resolve the weak convergence of the mixture measure, using both \cref{thm:exact asymptotics} and \cref{thm:weight split}, we compute
\begin{align*}
\frac{\kappa_n^g (\overline{B}(m' - \delta, m' + \delta))}{ \sum_{m^* \in \mathcal{M}^* (\psi^g)} \kappa_n^g (\overline{B}(m^* - \delta, m + \delta))} =  \frac{n^{- \left( \frac{1}{2 k (m')} - \frac{1}{2 k_\infty} \right)} W_n^g (m', \delta)}{\sum_{m^* \in \mathcal{M}^* (\psi^g)} n^{- \left( \frac{1}{2 k (m^*)} - \frac{1}{2 k_\infty} \right)} W_n^g (m^*, \delta)} ,
\end{align*}
from which it follows that
\begin{align*}
\lim_{n \to \infty} \frac{\kappa_n^g  (\overline{B}(m' - \delta, m' + \delta))}{ \sum_{m^* \in \mathcal{M}^* (\psi^g)} \kappa_n^g (\overline{B}(m^* - \delta, m + \delta))} = \begin{cases} \frac{W^g(m')}{\sum_{m^* \in \mathcal{M}_\infty (\psi^g)} W^g ( m^*)}, \ &k(m') = k_\infty (\psi^g) \\ 0, \ &k(m') < k_\infty (\psi^g) \end{cases}
\end{align*}
Following this computation, we have the following result.
\begin{theorem} \label{thm:mixture weak full } Let $g \in C_b ([-1,1])$ be an infinitely continuously differentiable function such that $\psi^g$ has finitely many global maximizing points $M^* (\psi^g) \subset (-1,1)$ of finite type.
\\
\\
Then, it follows that
\begin{align*}
\lim_{n \to \infty} \kappa_n^g = \left( \sum_{m^* \in M_\infty (\psi^g)} W^g (m^*)  \right)^{-1} \sum_{m^* \in M_\infty^* (\psi^g)} W^g (m^*) \delta_{m^*} .
\end{align*}
\end{theorem}
\noindent
To finish, we can directly compute the following
\begin{align*}
\partial_2^2 [\psi^g] (m^*, 0, 0) = - \frac{2 \sqrt{\frac{1 + m^*}{2}}}{\sqrt{\frac{1 + m^*}{2}} + \sqrt{\frac{1 - m^*}{2}}}, \ - \partial_3^2 [\psi^g] (m^*, 0, 0) = - \frac{2 \sqrt{\frac{1 - m^*}{2}}}{\sqrt{\frac{1 + m^*}{2}} + \sqrt{\frac{1 - m^*}{2}}} ,
\end{align*}
this implies that the integral containing these terms does not depend on $g$, other than through the value of the global maximizing point. In addition, it is immediate that the factor $e^{\psi^g (m^*) - g (m^*)}$ does not depend on $g$ either. Furthermore, we immediately have
\begin{align*}
\int_{-\infty}^\infty d m \ e^{\frac{1}{(2k)!} \partial^{2k} [\psi^g] (m^*) m^{2k}} = \frac{1}{|\partial^{2k} [\psi^g] (m^*)|^\frac{1}{2k}} \int_{-\infty}^\infty dm \ e^{- \frac{m^{2k}}{(2k)!}} .
\end{align*}
We can thus combine all factors not depending functionally on $g$ into a single function $C^k : (-1,1) \to (0, \infty)$ given by
\begin{align*}
C^k(m^*) := \frac{e^{g(m^*)}}{e^{{\psi^g} (m^*)} \sqrt{1 - {m^*}^2}} \int_{\mathbb{R}^3} d \theta_1 d \theta_2 d m \ e^{\frac{1}{2} \partial_2^2 [{\psi^g}] (m^*, 0, 0) \theta_1^2 + \frac{1}{2} \partial_{3}^2 [{\psi^g}] (m^*, 0, 0) \theta_2^2 - \frac{m^{2k}}{(2k)!}} , 
\end{align*}
so that
\begin{align*}
W^g(m^*) = \frac{C^k (m^*)}{|\partial^{2k} [\psi^g] (m^*)|^\frac{1}{2k}} .
\end{align*}
Using \cref{thm:locally uniform convergence expectations}, \cref{thm:mixture weak full }, \cref{thm:general convergence}, and the form of the weights $W^g (m^*)$ given above, we have the final result.
\begin{theorem} \label{thm:limiting state full} Let $g \in C_b ([-1,1])$ be an infinitely continuously differentiable function such that $\psi^g$ has finitely many global maximizing points $M^* (\psi^g) \subset (-1,1)$ of finite type, and let $k_\infty := k_\infty (\psi^g)$. 
\\
\\
Then, it follows that
\begin{align*}
\lim_{n \to \infty} \mu_n^g = \left( \sum_{m^* \in M_\infty^* (\psi^g)} \frac{C^{k_\infty} (m^*)}{|\partial^{2 k_\infty} [\psi^g] (m^*)|} \right)^{-1}\sum_{m^* \in M_\infty^* (\psi^g)} \frac{C^{k_\infty} (m^*)}{|\partial^{2 k_\infty} [\psi^g] (m^*)|} \eta (\beta(m^*,1), \mu (m^*, 1)) .
\end{align*}
\end{theorem}
\section{Intermediate results and proofs} \label{sec:Intermediate results and proofs}
\noindent
This section contains proof of some of the results in \cref{sec:Main results}, and some collections of intermediate results and theory that are required.
\subsection{Microcanonical probability measures} \label{sec:Microcanonical probability measures}
\noindent
To motivate the rigorous definition of the microcanonical ensemble and its associated probability measure, consider the following formal calculation
\begin{align*}
&\int_{\mathbb{R}^n} d \phi \ \delta (M_n (\phi) - mn) \delta(N_n(\phi) - \rho n) f (\phi)  \\ &= \sum_{\sigma \in \{ -1,1\}^n} \int_{[0, \infty)^n} d \phi \  \\ &\times \delta \left( \sum_{i \in \sigma^{-1} \{ +1 \}} \phi_i - \sum_{i \in \sigma^{-1} \{ -1 \}} \phi_i - mn \right) \delta \left( \sum_{i \in \sigma^{-1} \{ +1 \}} \phi_i + \sum_{i \in \sigma^{-1} \{ -1 \}} \phi_i - \rho n \right) f (\sigma \phi) \\
&= \sum_{\sigma \in \{ -1,1\}^n} \frac{1}{2} \int_{[0, \infty)^n} d \phi \  \\ &\times \delta \left( \sum_{i \in \sigma^{-1} \{ +1 \}} \phi_i - \frac{\rho + m}{2}n \right) \delta \left( \sum_{i \in \sigma^{-1} \{ -1 \}} \phi_i - \frac{\rho - m}{2} n \right) f (\sigma \phi) ,
\end{align*}
where the pair $(m, \rho) \in \mathcal{A}$ , $f : \mathbb{R}^n \to \mathbb{R}$ is a sufficiently regular function, and $\sigma \phi$ notation for a multiplication map defined  by $(\sigma \phi)_i := \sigma_i \phi_i$. Note that the integral in the sum is a product of two integrals since the index sets $\sigma^{-1} \{ +1 \}$ and $\sigma^{-1} \{ -1 \}$ are trivially disjoint. Note that the primary formal rule we have made use of is the following one
\begin{align*}
\delta(T x - y) = \frac{1}{|\det(T)|} \delta (x - T^{-1} y)
\end{align*}
for an invertible linear map $T : \mathbb{R}^k \to \mathbb{R}^k$, and elements $x,y \in \mathbb{R}^k$.
\\
\\
To make this formal calculation rigorous, we need to define integrals over scaled simplexes in arbitrary dimensions. To do this, we introduce the so-called flag coordinates $\phi' : \mathbb{R}^k \to \mathbb{R}^k$ given by
\begin{align*}
\phi_i' (\phi) := \sum_{j=1}^i \phi_i .
\end{align*}
Note that $\phi' ([0, \infty)^k) = \{ \phi \in [0, \infty)^k : \phi_1 \leq \phi_2 \leq ... \leq \phi_k \}$, $\det (\phi') = 1$, and the inverse function of $\phi'$ is given by 
\begin{align*}
{\phi'}^{-1}_i (\phi') = \phi'_{i} - \phi'_{i-1} ,
\end{align*}
where we take the convention that $\phi'_0 := 0$.
\\
\\
The connection between the flag coordinates and the integrals over simplexes can be seen from the following formal calculation
\begin{align*}
&\int_{[0, \infty)^k} d \phi \ \delta  \left( \sum_{i=1}^k \phi_i - r \right) f (\phi) \\ &= \int_{[0, \infty)^k} d \phi \ \delta (\phi_k - r) \mathbbm{1}(\phi_1 \leq \phi_2 \leq ... \leq \phi_k) f (\phi_1, \phi_2 - \phi_1,..., \phi_k - \phi_{k-1}) \\
&= \int_{[0, \infty)^{k-1}} d \phi \ \mathbbm{1}(\phi_1 \leq \phi_2 \leq ... \leq \phi_{k-1} \leq r) f (\phi_1, \phi_2 - \phi_1,..., r - \phi_{k-1}) ,
\end{align*}
where $r > 0$, and $f : \mathbb{R}^n \to \mathbb{R}$ is a sufficiently regular function.
\\
\\
From this formal calculation, we produce the following definition. 
\begin{definition}\label{def:uniform simplex integral} For a finite index set $I$ and $r > 0$, the measure $S_I (r)$ on $[0, \infty)^I$ corresponding to the integral over an $(|I|-1)$-dimensional $r$-scaled simplex on the index set $I$ is given by its action on $f \in C_b ([0, \infty)^I)$ given by
\begin{align*} 
S_I (r) [f] := \int_{[0, \infty)^{k-1}} d \phi \ \mathbbm{1}(\phi_{i_1} \leq \phi_{i_2} \leq ... \leq \phi_{i_{|I|-1}} \leq r) f (\phi_{i_1}, \phi_{i_2} - \phi_{i_1}\,..., r - \phi_{i_{|I| - 1}}) ,
\end{align*}
where $\{ i_k \}_{k=1}^{|I|}$ is some enumeration of $I$.
\end{definition}
\noindent
For future use, whenever it is clear that we are either referring to the measure or the normalization constant, we will use the following notation
\begin{align*}
S_I (r) := S_I (r) [1] = \frac{r^{|I|-1}}{(|I|-1)!} ,
\end{align*}
where the right-hand side follows by direct computation. Using dominated convergence, it is also clear that the mapping $r \mapsto S_I (r) [f]$ is continuous if $f \in C_b ([0, \infty)^I)$ is continuous. 
\\
\\
To show that \cref{def:uniform simplex integral} is independent of the enumeration of $I$ given above, we will use a Lebesgue-absolutely continuous approximation of $S_I (r)$. Let $g : [0, \infty) \to \mathbb{R}$ be a measurable function such that
\begin{align*}
\int_0^\infty dr \ |g(r)| r^{|I| - 1} < \infty .
\end{align*}
It follows that
\begin{align*}
\int_{[0, \infty)^{I}} d \phi \ g \left( \sum_{i \in I} \phi_i \right) f (\phi) = \int_0^\infty dr \ g (r) S_I (r) [f] ,
\end{align*}
where $f \in C_b ([0, \infty)^I)$. Now, consider the family $\{ g_\varepsilon \}_{\varepsilon > 0}$ given by
\begin{align} \label{def:approximation functions}
g_\varepsilon (r) := \frac{\mathbbm{1}(|r| < \varepsilon)}{2 \varepsilon} .
\end{align}
Fix $r > 0$. Since $f \in C_b ([0, \infty)^I)$, as stated before, one can verify that $S_I (\cdot) [f] \in C ([0, \infty))$. It follows that
\begin{align*}
S_I (r) [f] = \lim_{\varepsilon \to 0^+} \int_0^\infty dr' \ g_\varepsilon (r' - r) S_I(r') [f] = \lim_{\varepsilon \to 0^+} \int_{[0, \infty)^{I}} d \phi \ g_\varepsilon \left( \sum_{i \in I} \phi_i  - r \right) f (\phi) .
\end{align*}
We see that the left-hand side of the above equality will inherit properties from the right-hand side limiting term. In particular, the measure given by its action on $f \in C_b ([0, \infty)^I)$ given by
\begin{align*}
f \mapsto \int_{[0, \infty)^{I}} d \phi \ g_\varepsilon \left( \sum_{i \in I} \phi_i \right) f (\phi)
\end{align*}
is independent of any enumeration of $I$, and it is label permutation invariant. It follows that the measure $S_I (r)$ is independent of the given enumeration in the definition, and it is label permutation invariant. 
\\
\\
We can now define the microcanonical probability measure using \cref{def:uniform simplex integral}.  
\begin{definition} The measure $Z_n (M, N)$ is given by its action on $f \in C_b (\mathbb{R}^n)$ given by
\begin{align*}
Z_n (M, N) [f] := \begin{cases} \frac{1}{2} \sum_{\sigma \in \{ -1,1\}^n}  \left( S_{\sigma^{-1} \{ +1 \}} \left( \frac{N+M}{2} \right) \otimes S_{\sigma^{-1} \{ -1 \}} \left( \frac{N-M}{2} \right) \right)  [f \circ \sigma],& \ (M,N) \in \mathcal{A}, \\
S_n (N) [f],& \ (M,N) \in \partial \mathcal{A} \setminus \{ 0 \} ,
\end{cases} 
\end{align*}
where  $\otimes (\cdot)$ is the tensor product of two measures, $f \circ \sigma$ is the composition of the multiplication map $\sigma$ with $f$, and we take the necessary convention that 
\begin{align*}
\left( S_{\sigma^{-1} \{ +1 \}} \left( \frac{N+M}{2} \right) \otimes S_{\sigma^{-1} \{ -1 \}} \left( \frac{N-M}{2} \right) \right)  [f \circ \sigma] = 0
\end{align*}
if $\sigma  = \{ 1,1,...,1\}$ or $\sigma = \{ -1,-1,...,-1\}$ whenever $(M,N) \in \mathcal{A}$. 
\end{definition}
\noindent
This last convention implies that we do not include the "first" and "last" in the sum, but we have left them in to save space on notation.
\\
\\
To conclude this section, we will, finally, give the definition of the microcanonical probability measure. 
\begin{definition} \label{def:microcanonical probability measure rigorous}For $(m, \rho) \in \overline{\mathcal{A}} \setminus \{ 0 \}$, the probability measure $\nu_n (m,\rho)$ on $\mathbb{R}^n$ corresponding to the microcanonial probability measure is defined by its action on $f \in C_b (\mathbb{R}^n)$ given by
\begin{align}
\nu_n (m,\rho) [f] := \frac{Z_n (mn,\rho n) [f]}{Z_n (mn, \rho n)} ,
\end{align}
and the microcanonical partition function, acting as the normalization constant $Z_n (mn, \rho n)$ is given by
\begin{align} \label{def:microcanonical partition function}
 Z_n (mn, \rho n) := Z_n (mn, \rho n) [1] = \begin{cases} \frac{1}{2} \sum_{k=1}^{n - 1} {n \choose k} \frac{\left( \frac{\rho n + m n}{2} \right)^{k - 1}}{(k - 1)!} \frac{\left( \frac{\rho n - m n}{2} \right)^{n - k - 1}}{(n - k - 1)!},& \ (m, \rho) \in \mathcal{A}, \\
\frac{(\rho n)^{n - 1}}{(n-1)!}, & \ (m, \rho) \in \partial \mathcal{A} \setminus \{ 0 \} , \end{cases} 
\end{align}
which can be verified by direct computation.
\end{definition}
\noindent
To make the microcanonical probability measure computationally tractable, we will utilize a similar Lebesgue-absolutely continuous approximation as for the integrals over the simplex. However, as opposed to the approximation for the integrals over the simplexes, one must be more careful here. Using the family of functions $\{ g_\varepsilon \}_{\varepsilon > 0}$ from \cref{def:approximation functions}, observe that
\begin{align*}
&\int_{[0, \infty)^{\sigma^{-1} \{ + 1 \}} \times [0, \infty)^{\sigma^{-1} \{ - 1 \}}} d \phi \ \\ &\times g_\varepsilon \left( \sum_{i \in \sigma^{-1} \{ +1 \}} \phi_i - \frac{\rho n + mn}{2} \right) g_\varepsilon \left( \sum_{i \in \sigma^{-1} \{ -1 \}} \phi_i - \frac{\rho n - mn}{2}  \right) f (\sigma \phi) \\
&=  \int_{[0, \infty)^{\sigma^{-1} \{ + 1 \}} \times (-\infty, 0]^{\sigma^{-1} \{ - 1 \}}} d \phi \ \\ &\times g_\varepsilon \left( \sum_{i =1 }^n \frac{|\phi_i| + \phi_i}{2} - \frac{\rho n + mn}{2} \right) g_\varepsilon \left( \sum_{i =1 }^n \frac{|\phi_i| - \phi_i}{2} - \frac{\rho n - mn}{2}  \right) f (\phi) ,
\end{align*}
where $(m, \rho) \in \mathcal{A}$, and $\sigma$ does not consist of all $1$'s or all $-1$'s. Now, if we consider instead the right-hand side first, then it makes sense even when $\sigma$ consists of all $1$'s or $-1$'s. In that instance, the argument of one of the $g_\varepsilon$ will not integrate over any $\phi$-variables, and for small enough $\varepsilon > 0$ the indicator function vanishes. Summing over the $\sigma$, in this case, it then follows that 
\begin{align} \label{def:fundamental approximation}
&Z_n (m n,  \rho n) [f] \\ &= \lim_{\varepsilon \to 0^+} \frac{1}{2} \int_{\mathbb{R}^n} d \phi \  g_\varepsilon \left( \sum_{i = 1}^n \frac{|\phi_i| + \phi_i}{2} - \frac{\rho n + mn}{2} \right) g_\varepsilon \left( \sum_{i =1}^n \frac{|\phi_i| - \phi_i}{2} - \frac{\rho n - mn}{2}  \right) f (\phi) . \notag
\end{align}
Returning now to the microcanonical probability measure, we see that its inherits the various properties of the measure with action on $f \in C_b (\mathbb{R}^n)$ given by
\begin{align*}
f \mapsto \int_{\mathbb{R}^n} d \phi \ g_\varepsilon \left( \sum_{i =1 }^n \frac{|\phi_i| + \phi_i}{2} - \frac{\rho n + mn}{2} \right) g_\varepsilon \left( \sum_{i =1 }^n \frac{|\phi_i| - \phi_i}{2} - \frac{\rho n - mn}{2}  \right) f (\phi) .
\end{align*}
In particular, it is label permutation invariant. Furthermore, this approximation will be used for some calculations related to the microcanonical probability measure.
\subsection{Relative entropy and local observables} \label{sec:Relative entropy and local observables}
We begin with the proof of the type of generalized dominated convergence theorem.
\begin{proof}[Proof of \cref{thm:general convergence}]
The condition that $K$ is a continuity set of $\mu$ implies that
\begin{align*}
\lim_{n \to \infty} \mu_n (K) = \mu (K),
\end{align*}
and the condition that $\operatorname{supp} (\mu) \subset K$ implies that $\mu (K) = 1$.
\\
\\
Next, we have the following two simple inequalities
\begin{align*}
\left| \int_X \mu_n (dx) \ f_n (x) - \int_K \mu_n (dx) \ f_n (x) \right| \leq \mu_n(X \setminus K) \sup_{n \in \mathbb{N}} \sup_{x \in X} |f_n(x)| ,
\end{align*}
and
\begin{align*}
\left| \int_K \mu_n (dx) \ f_n (x) - \int_K \mu_n (dx) \ f (x) \right| \leq \mu_n (K) \sup_{x \in K} |f_n (x) - f(x)| .
\end{align*}
Since $K$ is a continuity set of $\mu$, using the continuity set definition of weak convergence, it follows that $\mu_n$ conditioned to $K$ converges weakly to $\mu$ conditioned to $K$. Transitioning to the continuous bounded form of weak convergence, it follows that
\begin{align*}
\lim_{n \to \infty} \frac{1}{\mu_n (K)} \int_K \mu_n (dx) \ f(x) = \frac{1}{\mu(K)}\int_K \mu (dx) \ f(x) = \int_K \mu(dx) \ f(x) .
\end{align*}
For completeness, we have the following final inequality
\begin{align*}
\left| \int_K \mu_n (dx) \ f (x) - \int_X \mu(dx) \ f(x) \right| \leq \frac{\mu_n (X \setminus K)}{\mu_n (K)} \left| \int_X \mu_n (dx) \ f(x) \right|.
\end{align*}
Combining together all three inequalities, the result follows.
\end{proof}
\noindent
We will need that the relative entropy between two absolutely continuous probability measures.
\begin{definition} Let $X$ be a Polish space, and let $\mu$ and $\nu$ be probability measures on $X$. If $\mu$ is absolutely continuous with respect to $\nu$, the relative entropy $\mathcal{H}(\mu || \nu)$ is given by
\begin{align*}
\mathcal{H}(\mu || \nu) := \int_X d \mu \ln \frac{d \mu}{d \nu} .
\end{align*}
If $\mu$ is not absolutely continuous with respect to $\nu$, we set $\mathcal{H} (\mu || \nu) = \infty$.
\end{definition} 
\noindent
We will need the following properties of relative entropy.
\begin{theorem} \label{thm:relative entropy properties}
Let $X$ be a Polish space, and let $\mu$ and $\nu$ be probability measures on $X$ such that $\mu$ is absolutely continuous with respect to $\nu$. 
\begin{itemize}
\item For any $\mu$ and $\nu$ satisfying the assumptions 
\begin{align*} 
\mathcal{H}(\mu || \nu ) \geq 0 .
\end{align*}
\item  For any $\mu$ and $\nu$ satisfying the assumptions
\begin{align*}
\sup_{f \in M_b (X), \ || f ||_\infty \leq 1} |\mu [f] - \nu [f]| \leq \sqrt{\frac{\mathcal{H}(\mu || \nu )}{2}}, 
\end{align*}
where $M_b (X)$ is the space of measurable bounded functions on $X$.
\item If $X = Y^n$, where $Y$ is another Polish space, and $\nu = \otimes_{k=1}^n \lambda$, where $\lambda$ is a probability measure on $Y$, it follows that
\begin{align*}
\mathcal{H}_I (\mu || \nu) + \mathcal{H}_J (\mu || \nu) \leq \mathcal{H}_{I \cup J} (\mu || \nu ) + \mathcal{H}_{I \cap J} (\mu || \nu) ,  
\end{align*}
where $I,J \subset \{ 1,2,...,n\}$, and $\mathcal{H}_I (\mu || \nu)$ is denotes the relative entropy of the $I$:th marginal distributions of $\mu$ and $\nu$. 
\end{itemize}
\end{theorem}
\noindent
The first and third properties properties are discussed and given proofs in \cite{Georgii2011}. The second property is sometimes referred to as Pinsker's inequality and references to proofs and other details concerning this inequality can be found in \cite{Sason2016}.
\\
\\
We can now give a proof of the fundamental inequality connecting the constrained and non-constrained ensemble probability measures.
\begin{proof}[Proof of \cref{thm:fundamental inequality}] Using \cref{def:fundamental approximation}, we can compute the integral over only the first $2$ variables leaving the other $n-2$ variables fixed. We compute
\begin{align*}
&\int_{\mathbb{R}^2} d \phi \ g_\varepsilon \left( \sum_{i = 1}^n \frac{|\phi_i| + \phi_i}{2} - \frac{\rho n + mn}{2} \right)  g_\varepsilon \left( \sum_{i = 1 }^n \frac{|\phi_i| - \phi_i}{2} - \frac{\rho n - mn}{2}  \right)  \\
&= g_\varepsilon \left( \sum_{i = 3 }^n \frac{|\phi_i| - \phi_i}{2} - \frac{\rho n - mn}{2}  \right) \int_{[0, \infty)^2} d \phi \ g_\varepsilon \left( \phi_1 + \phi_2 - \left( \frac{\rho  n + m n}{2} - \sum_{i=3}^n \frac{|\phi_i| + \phi_i}{2} \right)\right)  \\
&+ 2 \int_{[0, \infty)^2} d \phi \ g_\varepsilon \left( \phi_1 - \left( \frac{\rho  n + m n}{2} - \sum_{i=3}^n \frac{|\phi_i| + \phi_i}{2} \right)\right) g_\varepsilon \left( \phi_2 - \left( \frac{\rho  n - m n}{2} - \sum_{i=3}^n \frac{|\phi_i| - \phi_i}{2} \right)\right) \\
&+   g_\varepsilon \left( \sum_{i = 3 }^n \frac{|\phi_i| + \phi_i}{2} - \frac{\rho n + mn}{2}  \right) \int_{[0, \infty)^2} d \phi \ g_\varepsilon \left( \phi_1 + \phi_2 - \left( \frac{\rho  n - m n}{2} - \sum_{i=3}^n \frac{|\phi_i| - \phi_i}{2} \right)\right) .
\end{align*}
Taking the limit, it follows that
\begin{align*}
&\lim_{\varepsilon \to 0^+} \int_{\mathbb{R}^2} d \phi \ g_\varepsilon \left( \sum_{i = 1}^n \frac{|\phi_i| + \phi_i}{2} - \frac{\rho n + mn}{2} \right)  g_\varepsilon \left( \sum_{i = 1 }^n \frac{|\phi_i| - \phi_i}{2} - \frac{\rho n - mn}{2}  \right) \\ &= 2  \mathbbm{1} \left( \frac{\rho  n + m n}{2} - \sum_{i=3}^n \frac{|\phi_i| + \phi_i}{2} \geq 0 \right) \mathbbm{1} \left( \frac{\rho  n - m n}{2} - \sum_{i=3}^n \frac{|\phi_i| - \phi_i}{2} \geq 0 \right) .
\end{align*}
Accounting for the normalization, the $(n-2)$:th marginal of the microcanonical probability measure is given by
\begin{align*}
\nu_n (m, \rho) (d \phi_{n-2}) = \frac{ \mathbbm{1} \left( \frac{\rho  n + m n}{2} - \sum_{i=3}^n \frac{|\phi_i| + \phi_i}{2} \geq 0 \right) \mathbbm{1} \left( \frac{\rho  n - m n}{2} - \sum_{i=3}^n \frac{|\phi_i| - \phi_i}{2} \geq 0 \right)}{Z_n (m n, \rho n)} d \phi_{n-2} ,
\end{align*}
where $d \phi_{n-2}$ is the $(n-2)$-dimensional Lebesgue measure. Note the factor of $2$ vanishes due to the presence of a factor of $\frac{1}{2}$ in the partition function. It follows that
\begin{align*}
\frac{d \nu_n (m, \rho)}{d \eta_n (\beta, \mu)} (\phi_{n-2}) &= \frac{Q_{n-2} (\beta, \mu)}{e^{-\beta \sum_{i=3}^{n} \phi_i - \mu \sum_{i=3}^n |\phi_i|}}  \\ &\times \frac{  \mathbbm{1} \left( \frac{\rho  n + m n}{2} - \sum_{i=3}^n \frac{|\phi_i| + \phi_i}{2} \geq 0 \right) \mathbbm{1} \left( \frac{\rho  n - m n}{2} - \sum_{i=3}^n \frac{|\phi_i| - \phi_i}{2} \geq 0 \right)}{Z_n ( m n, \rho n)} d \phi_{n-2} .
\end{align*}
The relative entropy is then directly computed to be
\begin{align*}
\mathcal{H}_{n-2} (\nu_n (m, \rho) || \eta_n (\beta, \mu)) = \beta \nu_{n} (m, \rho) \left[ M_{n-2} \right] + \mu \nu_{n} (m, \rho) \left[ N_{n-2} \right] + \ln Q_{n - 2} (\beta, \mu) - \ln Z_n (m n, \rho n) . 
\end{align*}
Using label permutation invariance, one can directly compute that
\begin{align*}
\nu_{n} (m, \rho) \left[ M_{n-2} \right] = (n - 2) m, \ \nu_{n} (m, \rho) \left[ N_{n-2} \right] = (n - 2) \rho , \ \ln Q_{n - 2} (\beta, \mu) = (n - 2) f(\beta, \mu) .
\end{align*}
In summary, we have
\begin{align*}
\frac{1}{n - 2} \mathcal{H}_{n-2} (\nu_n (m, \rho) || \eta_n (\beta, \mu)) = \beta m + \mu \rho + f (\beta, \mu) - \frac{n}{n - 2} s_n (m, \rho) .
\end{align*}
To continue, by \cref{thm:relative entropy properties}, it follows that
\begin{align*}
\sup_{f \in C_b (\mathbb{R}^I), \ || f ||_\infty \leq 1} |\nu_n (m, \rho) [f] - \eta_n (\beta, \mu) [f]| \leq \sqrt{\frac{\mathcal{H}_I (\nu_n (m, \rho) || \eta_n (\beta, \mu))}{2}} .
\end{align*}
By label permutation invariance, it follows that 
\begin{align*}
\mathcal{H}_I (\nu_n (m, \rho) || \eta_n (\beta, \mu)) = \mathcal{H}_{[|I|]} (\nu_n (m, \rho) || \eta_n (\beta, \mu)) .
\end{align*}
Since $I$ is finite, it follows that there exists $k \in \mathbb{N}$ such that $(k - 1) |I| \leq n - 2 < k |I|$. Since $\eta_n (\beta, \mu)$ is a product measure, using \cref{thm:relative entropy properties}, it follows that
\begin{align*}
\mathcal{H}_I (\nu_n (m, \rho) || \eta_n (\beta, \mu)) &= \frac{1}{k - 1} \sum_{j=1}^{k - 1} \mathcal{H}_{[|I|] + (j-1) |I|} (\nu_n (m, \rho) || \eta_n (\beta, \mu)) \\ &\leq \frac{\mathcal{H}_{n-2} (\nu_n (m, \rho) ||\eta_n (\beta, \mu))}{k - 1} \\
&\leq |I| \frac{\mathcal{H}_{n-2} (\nu_n (m, \rho) || \eta_n (\beta, \mu))}{n - 2 - |I|} .
\end{align*}
Combining these inequalities together, it follows that
\begin{align*}
&\sup_{f \in C_b (\mathbb{R}^I), \ || f ||_\infty \leq 1} |\nu_n (m, \rho) [f] - \eta_n (\beta, \mu) [f]| \\ &\leq \sqrt{\frac{|I| (n - 2)}{2(n - 2 - |I|)} \left( \beta m + \mu \rho + f(\beta, \mu) - \frac{n}{n - 2} s_n (m, \rho) \right)} ,
\end{align*}
as desired. 
\end{proof}
\subsection{Large deviations and weak convergence} \label{sec:Large deviations and weak convergence}
\noindent
We begin with the standard key definitions of large deviations theory. Note that these definitions are either the same or slightly modified versions of the same results and definitions found in \cite{Hollander2008}. In addition, the result concerning convexity are either provided in \cite{Hollander2008}, or we refer to \cite{Rockafellar1997} for more detailed analysis of convex objects.
\\
\\
 In the following $\{ P_n \}_{n = 1}^\infty$ is a sequence of probability measures on a Polish space $X$.
\begin{definition} A function $I : X \to [0, \infty]$ is called a rate function if it satisfies the following properties
\begin{itemize}
\item $I(x) < \infty$ for all $x \in X$.
\item $I$ is lower semi-continuous.
\item $I$ has compact level sets.
\end{itemize}
\end{definition}
\noindent 
In the following, we use the notation $I(A) := \inf_{x \in A} I(x)$.
\begin{definition} A sequence of probability measure $\{ P_n \}_{n = 1}^\infty$ is said to satisfy a large deviations principle with rate function $I$ if it satisfies the following properties
\begin{itemize}
\item For all closed sets $C \subset X$, we have
\begin{align*}
\limsup_{n \to \infty} \frac{1}{n} \ln P_n (C) \leq - I (C) .
\end{align*}
\item For all open sets $O \subset X$, we have
\begin{align*}
\liminf_{n \to \infty} \frac{1}{n} \ln P_n (O) \geq - I (O) .
\end{align*}
\end{itemize}
\end{definition}
\noindent
Now, we specialize to probability distributions on $\mathbb{R}^d$. In the following, let $\{ m_n \}_{n=1}^\infty$ be a sequence of random variables on $\mathbb{R}^d$, and we set $P_n (A) := \mathbb{P}(m_n \in A)$. The moment generating functions $\varphi_n  : \mathbb{R}^d \to (0, \infty]$ are given by $\varphi_n (t) := \mathbb{E} e^{ \left< t, m_n \right>}$. In the following, we assume the existence of a function $\Lambda : \mathbb{R}^d \to [- \infty, \infty]$ given by
\begin{align*}
\Lambda (t) := \lim_{n \to \infty} \frac{1}{n} \ln \varphi_n (n t) ,
\end{align*}
and that this function satisfies $0 \in \operatorname{int}(\mathcal{D} (\Lambda))$ where $\mathcal{D}(\Lambda) := \{ t \in \mathbb{R}^d : \Lambda(t) < \infty \}$. For such a function, it follows that $\Lambda$ is convex and $\Lambda(t) > - \infty$ for all $t \in \mathbb{R}^d$. A convex function $\Lambda : \mathbb{R}^d \to [- \infty, \infty]$ is called proper if $\Lambda (t) > - \infty$ for all $t \in \mathbb{R}^d$, and there exists at least one point $t_0 \in \mathbb{R}$ such that $\Lambda (t_0) < \infty$. It is clear that when $\Lambda$ is the limit of the scaled logarithmic moment generating functions, then it is a proper convex function.
\noindent
We will need the Legendre transform of $\Lambda$.
\begin{definition} The Legendre transform $\Lambda^* : \mathbb{R}^d \to [- \infty, \infty]$ of a $\Lambda : \mathbb{R}^d \to [-\infty, \infty]$ is given by
\begin{align*}
\Lambda^* (x) := \sup_{t \in \mathbb{R}^d} \{ \left< x, t \right> - \Lambda (t) \} .
\end{align*} 
\end{definition}
\noindent
For $\Lambda$ given by the limit of the scaled logarithm, it follows that $\Lambda^*$ is a convex rate function. In particular, we see that the range of $\Lambda^*$ must be contained in $[0, \infty)$. 
\\
\\
To specify the form of the Gärtner-Ellis theorem, that we wish to utilize, we need the concept of essential smoothness.
\begin{definition} A proper convex function $\Lambda : \mathbb{R}^d \to (-\infty, \infty]$ is called essentially smooth if it satisfies the following properties
\begin{itemize}
\item $\operatorname{int} (\mathcal{D} (\Lambda)) \not = \emptyset$.
\item $\Lambda$ is differentiable on $\operatorname{int} (\mathcal{D} (\Lambda))$.
\item Either $\mathcal{D}(\Lambda) = \mathbb{R}^d$ or, for any $t^* \in \partial \mathcal{D}(\Lambda)$, it follows that $\lim_{t \to t^*} || \nabla [\Lambda] (t)|| = \infty$.
\end{itemize} 
\end{definition}
\noindent
We can now give the essentially smooth form of the Gärtner-Ellis theorem.
\begin{theorem}
Let $\Lambda : \mathbb{R}^d \to (0, \infty]$ be an essentially smooth lower semi-continuous function.
\\
\\
It follows that $\{ P_n \}_{n=1}^\infty$ satisfies a large deviations principle with rate function $\Lambda^*$. 
\end{theorem}
\noindent
It is typical to introduce the notion of strict convexity of a function, but we will instead directly introduce the notion of a Legendre-type function.
\begin{definition} A proper convex lower semi-continuous function $\Lambda : \mathbb{R}^d \to (-\infty, \infty]$ is said to be of Legendre-type if it is both essentially smooth and strictly concave on $\operatorname{int} (\mathcal{D} (\Lambda))$. 
\end{definition}
\noindent
The primary feature of Legendre-type functions that we will use is that the gradient of such a function $\Lambda$ is a bijection between $\operatorname{int} (\mathcal{D} (\Lambda))$ and $\operatorname{int} (\mathcal{D} (\Lambda^*))$. 
\\
\\
We can now prove the following general theorem. 
\begin{theorem}
Let $\{ Z_n \}_{n=1}^\infty$ be a sequence of functions $Z_n : n \mathcal{A} \to (0, \infty)$, where $\mathcal{A} \subset \mathbb{R}^d$ is a non-empty open convex set such that each $Z_n$ is log-concave, and 
\begin{align*}
\sup_{n \in \mathbb{N}} \left|\frac{1}{n} \ln Z_n (x n) \right| < \infty
\end{align*}
for each $x \in \mathcal{A}$. Denote by $s_n : \mathcal{A} \to (- \infty, \infty)$ the function given by
\begin{align*}
s_n (x) := \frac{1}{n} \ln Z_n (x n) .
\end{align*} 
In addition, suppose that the function $f : \mathbb{R}^d \to [- \infty, \infty]$ given by
\begin{align*}
f(t) := \lim_{n \to \infty} \frac{1}{n} \ln Q_n (t),
\end{align*}
exists, where $Q_n : \mathbb{R}^d \to (0, \infty]$ are given by
\begin{align*}
Q_n (t) := \int_{n \mathcal{A}} d X \ e^{-  \left< t, X \right>} Z_n (X) ,
\end{align*}
and there exists a non-empty open convex set $\mathcal{B} \subset \mathbb{R}^d$ such that $\mathcal{D} (Q_n) = \operatorname{int} (\mathcal{D} (f)) = \mathcal{B}$. 
\\
\\
If $f$ is a proper convex lower semi-continuous function of Legendre type such that $-\nabla [f] \mathcal{B} = \mathcal{A}$ then the function $s : \mathcal{A} \to \mathbb{R}$ given by the limit
\begin{align*}
s(x) := \lim_{n \to \infty} \frac{1}{n} \ln Z_n (x n) ,
\end{align*}
exists, and satisfies
\begin{align*}
s(x) := \inf_{t \in \mathbb{R}^d} \{ \left< t, x\right> + f(t)\}, \ \lim_{n \to \infty} \sup_{K \subset \mathcal{A}} |s_n (x) - s (x)| = 0 
\end{align*}
for any compact set $K \subset \mathcal{A}$. 
\end{theorem}
\begin{proof}
For the first step, let $t_0 \in \mathcal{B}$ be any base point, and we define the sequence of probability measures $\{ P_n \}_{n = 1}^\infty$ on $\mathbb{R}^d$ by setting
\begin{align*}
P_n (A) := \frac{1}{Q_n (t_0)} \int_{n (A \cap \mathcal{\mathcal{A}})} dX \ e^{- \left< t_0, X \right>} Z_n (X) ,
\end{align*}
where $A \subset \mathbb{R}^d$ is Borel measurable.   
\\
\\
The moment generating function $\varphi_n : \mathbb{R}^d \to (0, \infty]$ of the random variable $m_n$ on $\mathbb{R}^d$ with distribution given by $P_n$ is given by
\begin{align*}
\varphi_n (t) := \frac{Q_n (t_0 - \frac{t}{n})}{Q_n (t_0)} .
\end{align*}
The limit of the scaled logarithm moment generating function $\Lambda : \mathbb{R}^d \to [- \infty, \infty]$ is given by
\begin{align*}
\Lambda (t) := \lim_{n \to \infty} \frac{1}{n} \ln \varphi_n ( n t) = f(t_0 - t) - f (t_0) . 
\end{align*}
Since $\Lambda$ inherits its properties from $f$, it follows that $\Lambda$ exists, is a proper convex lower semi-continuous function of Legendre-type, and satisfies $0 = t_0 - t_0 \in \operatorname{int} (\mathcal{D} (\Lambda)) = t_0 - \operatorname{int} (\mathcal{D}(f)) = t_0 - \mathcal{B}$. It follows that $\{ P_n \}_{n \in \mathbb{N}}$ satisfies a large deviations principle with rate function $\Lambda^*$. Since $\Lambda$ is of Legendre-type, it follows that $\operatorname{int} (\mathcal{D} (\Lambda^*)) = \nabla [\Lambda] (\operatorname{int} (\mathcal{D} (\Lambda))) = - \nabla [f] \mathcal{B} = \mathcal{A}$.
\\
\\
Let $y \in \operatorname{int} (\mathcal{D} (\Lambda^*)) = \mathcal{A}$. Since $\Lambda^*$ is convex, it follows that it is continuous on $\mathcal{A}$ and thus the compact balls $\overline{B}(y, \delta)$ for small enough $\delta > 0$ are continuity sets from which it follows that
\begin{align*}
\lim_{n \to \infty} \frac{1}{n} \ln P_n (\overline{B} (y, \delta)) = - \Lambda^* (\overline{B}(y, \delta)) .
\end{align*}
For the second step, since each $s_n$ is concave and the collection $\{ s_n \}_{n \in \mathbb{N}}$ is pointwise uniformly bounded, it follows that the collection $\{ s_n \}_{n \in \mathbb{N}}$ is relatively compact in the compact-open topology of  continuous functions. Let $\{ s_{n_k} \}_{k = 1}^\infty$ be any locally uniformly convergent subsequence with limiting function $s'$. Since $\overline{B}(y, \delta)$ is a compact set, it follows that
\begin{align*}
\lim_{k \to \infty} \frac{1}{n_k} \ln \int_{n_k \overline{B}(y, \delta)} d X \ e^{- \left< t_0, X \right>} Z_{n_k} (X) &= \lim_{k \to \infty} \frac{1}{n_k} \ln \int_{ \overline{B}(y, \delta)} d x \ n_k e^{n_k \left( s_{n_k} (x) -  \left< y_0, x\right> \right)}  \\
&= \sup_{x \in \overline{B}(y, \delta)} \{ s' (x) - \left< t_0, x \right> \} .
\end{align*}
Then we have
\begin{align*}
\lim_{k \to \infty} \frac{1}{n_k} \ln P_{n_k} (\overline{B}(y, \delta)) = \sup_{x \in \overline{B} (y, \delta)} \{ s' (x) - \left< t_0, x \right>\} - f (t_0) .
\end{align*}
By combining this result with the large deviations principle, we deduce that
\begin{align*}
\sup_{x \in \overline{B} (y, \delta)} \{ s' (x) - \left< t_0, x \right>\} - f (t_0) = - \Lambda^* (\overline{B}(y, \delta)) .
\end{align*}
Now, since both functions inside the supremum and infimum respectively are continuous, letting $\delta \to 0^+$, we obtain
\begin{align*}
s' (y) - \left< \beta_0, y \right> - f (t_0) = - \Lambda^* (y) \iff s' (y) =  \inf_{t \in \mathbb{R}^d} \{ \left< y,t \right> + f (t) \} .
\end{align*}
Since $s'$ was the locally uniform limit of an arbitrary convergent subsequence $\{ s_{n_k }\}_{k = 1}^\infty$, the above result implies that this holds for any such $s'$, and thus the limit of any convergent subsequence is the same from which it follows that
\begin{align*}
\lim_{n \to \infty} s_n (x) = \inf_{t \in \mathbb{R}^d} \{ \left< t, x \right> + f (t) \},
\end{align*}
for $x \in \operatorname{int} (\mathcal{D} (\Lambda^*)) = \mathcal{A}$, and since the $s_n$ are concave and pointwise uniformly bounded, this convergence is automatically locally uniform. 
\end{proof}
\noindent
Let us also give a quick proof of the following weak convergence result concerning large deviations principles.
\begin{theorem} \label{thm:limiting points large deviations}
Let $\{ P_n \}_{n = 1}^\infty$ be a sequence of probability measures on $X$ satisfying a large deviations principle with rate function $I$.
\\
\\
It follows that
\begin{align*}
L \left( \{ P_n \}_{n=1}^\infty \right) \subset \{ P \in \mathcal{P} (X) : \operatorname{supp} (P) \subset I^{-1} \{ 0 \} \} ,
\end{align*}
where $\mathcal{P} (X)$ is the space of Borel probability measures on $X$.
\end{theorem}
\begin{proof}
Let us first show that $I^{-1} \{ 0 \}$ is non-empty and closed. Since $I$ has compact level sets, it follows that $I^{-1} [0,c]$ are compact for $c > 0$, but possibly empty. If they are not empty, then $I^{-1} \{ 0 \} = \bigcap_{n=1}^\infty I^{-1} \left[ 0, \frac{1}{n}\right]$, and it follows directly that $I^{-1} \{ 0 \}$ is non-empty and compact. However, if $I^{-1} [0, c]$ is empty for some $c > 0$, observe that
\begin{align*}
0 = \lim_{n \to \infty} I (I^{-1} [0,n]) = \lim_{n \to \infty} \inf_{x \in I^{-1} [0,n]} I(x) =  \lim_{n \to \infty} \inf_{x \in I^{-1} (c,n]} I(x) > c > 0 ,
\end{align*}
which is a contradiction, and thus $I^{-1} [0,c]$ are non-empty for every $c > 0$, and subsequently $c = 0$. Note that the first line of the above proof by contradiction follows from the fact that $\{ P_n \}_{n = 1}^\infty$ satisfies a large deviations principle. 
\\
\\
Let $y \not \in I^{-1} \{ 0 \}$ be such that $\overline{B}(y, \delta)$ is disjoint from $I^{-1} \{ 0 \}$ for small enough $\delta > 0$. Note that
\begin{align*}
\limsup_{n \to \infty} \frac{1}{n} \ln P_n (\overline{B}(y, \delta)) \leq - I (\overline{B}(y, \delta)) < 0 .
\end{align*}
The last strict inequality follows since by lower semi-continuity $I$ attains its minimum on any non-empty compact set, and $I$ is strictly positive on the set $\overline{B}(y, \delta)$. It follows that
\begin{align*}
P_n (B (y, \delta)) \leq P_n (\overline{B}(y, \delta)) \leq e^{n \sup_{k \geq n} \frac{1}{n} \ln P_n (\overline{B}(y, \delta))}, 
\end{align*}
so that
\begin{align*}
\lim_{n \to \infty} P_n (B(y, \delta)) = 0 .
\end{align*}
Since the sequence of probability measures satisfies a large deviations principle, it is exponentially tight which implies that it is uniformly tight in the weak sense. Let $\{ P_{n_k} \}_{k = 1}^\infty$ be any weakly convergent subsequence with limiting probability measure $P$. Let $\overline{B} (y, \delta)$ be as before, by weak convergence, it follows that
\begin{align*}
P (B(y, \delta)) \leq \liminf_{k \to \infty} P_{n_k} (B(y, \delta)) = 0 .
\end{align*}
Since $y \not \in I^{-1} \{ 0 \}$ is arbitrary, it follows that
\begin{align*}
\left( I^{-1} \{ 0 \} \right)^c \subset (\operatorname{supp} (P))^c \iff \operatorname{supp} (P) \subset I^{-1} \{ 0 \}  .
\end{align*}
\end{proof}
\noindent
For the purposes of this paper, the most important corollary is the case where $I^{-1} \{ 0 \}$ consists of a single point.
\begin{corollary} Let $\{ P_n \}_{n=1}^\infty$ be a sequence of probability measures on $X$ satisfying a large deviations principle with rate function $I$ such that $I(x^*) = 0$ for exactly one $x^* \in X$.
\\
\\
It follows that
\begin{align*}
\lim_{n \to \infty} P_n = \delta_{x^*}
\end{align*}
weakly. 
\end{corollary}
\noindent
The proof of this statement is an application of the previous theorem in combination with Prokhorov's theorem. 
\\
\\
Another important corollary is the following result concerning the case where $I^{-1} \{ 0 \}$ consists of finitely many points.
\begin{corollary} \label{thm:weight split}Let $\{ P_n \}_{n=1}^\infty$ be a sequence of probability measures on $X$ satisfying a large deviations principle with rate function $I$ such that the set $M^* := I^{-1} \{ 0 \}$ is finite.
\\
\\
It follows that
\begin{align*}
\int_X P_n (dx) \ f(x) = \sum_{x^* \in M^*} \frac{P_n (\overline{B}(x^*, \delta))}{P_n (A_\delta)} f (x^*) + o (1) 
\end{align*}
for any $0 < \delta < \min_{x^*, y^* \in M^*} d(x^*, y^*)$.
\end{corollary}
\begin{proof}
Let $\delta < \min_{x^*, y^* \in M^*} d(x^*,y^*)$. We decompose $X$ as follows
\begin{align*}
X = A_\delta \cup  A^c_\delta ,
\end{align*}
where
\begin{align*}
A_\delta := \bigcup_{x^* \in M^*} \overline{B}(x^*, \delta) .
\end{align*}
Using this decomposition, we have
\begin{align*}
P_n = P_n (A_\delta) \sum_{x^* \in M^*} \frac{P_n( \overline{B}(x^*, \delta))}{P_n(A_\delta)} P_n |_{\overline{B}(x^*, \delta)} + P_n (A_\delta^c) P_n |_{A_\delta^c} . 
\end{align*}
Using the large deviations principle, it follows that
\begin{align*}
\lim_{n \to \infty} P_n (A_\delta) = 1 ,
\end{align*}
and, by using the previous corollary, it follows that
\begin{align*}
\lim_{n \to \infty} P_n|_{\overline{B}(x^*, \delta)} = \delta_{x^*}
\end{align*}
weakly, where $x^* \in M^*$. Using these limits together, we obtain
\begin{align*}
\lim_{n \to \infty} \left| \int_X P_n (dx) \ f(x) - \sum_{x^* \in M^*} \frac{P_n (\overline{B}(x^*, \delta))}{P_n (A_\delta)} f (x^*) \right| = 0 .
\end{align*}
\end{proof}
\subsection{Infinite volume entropies and states} \label{sec:Infinite volume entropies and states}
Next, we prove the regularity and boundedness of the finite volume entropies.
\begin{proof}[Proof of \cref{thm:microcanonical entropy regularity}]
From \cref{def:microcanonical partition function}, we see that the microcanonical partition function is a homogeneous bivariate polynomial of degree $n-2$. Let us introduce the change of coordinates $z : \mathcal{A} \to (0, \infty)^2$ given by
\begin{align*}
z (M,N) := (x(M,N), y (M,N)) = \left( \frac{N + M}{2}, \frac{N - M}{2} \right) .
\end{align*}
It follows that $Z_n (M,N) = \frac{1}{2} P_n (z (M,N))$, where $P_n : (0, \infty)^2 \to (0, \infty)$ is given by
\begin{align*}
P_n (x,y) := \sum_{k=1}^{n - 1} {n \choose k} \frac{x^k}{(k-1)!} \frac{y^{n - k - 1}}{(n - k - 1)!} .
\end{align*}
Using the properties of the binomial coefficient, we can manipulate $P_n$ into the following form 
\begin{align*}
P_n (x,y) = n (n - 1) \sum_{k=0}^{n - 2} {n - 2 \choose k} \frac{x^{k}}{(k+1)!} \frac{y^{n - 2 - k}}{(n - 1  - k)!} .
\end{align*}
Let us denote the coefficients of the above manipulated polynomial by $\{ c_k \}_{k=0}^{n - 2}$. For $k \in \mathbb{N}$, using the simple relation
\begin{align*}
(k + 1)! (k - 1)! > k!, 
\end{align*}
it follows that
\begin{align*}
\frac{c_k^2}{{n - 2 \choose k}^2} > \frac{c_{k+1}}{{n - 2 \choose k + 1}} \frac{c_{k - 1}}{{n - 2 \choose k - 1}}
\end{align*}
for $0 < k < n - 2$. Using \cite[Example 2.3]{Braenden2020}, this implies that the sequence of coefficients $\{ c_k \}_{k=0}^{n - 2}$ is ultra log-concave, which yields that $P_n$ is Lorentzian, which shows that $P_n$ is log-concave, see \cite[Theorem 2.30]{Braenden2020} and the definition of completely log-concave polynomials due to \cite{Anari2021}. Since $Z_n$ is the composition of an invertible linear map, simple scaling by a factor of $2$, and a log-concave polynomial it follows that $Z_n$ is log-concave.
\\
\\
For boundedness, by \cref{thm:relative entropy properties}, we have $\mathcal{H}_{n - 2} (\nu_n (m, \rho) || \eta_n (\beta, \mu)) \geq 0$, from which it follows that
\begin{align*}
s_n (m, \rho) \leq \frac{n}{n - 2} s_n (m, \rho) \leq f(\beta, \mu) + \beta m + \mu \rho ,
\end{align*}
which shows that the family of entropies is pointwise bounded above. As for a lower bound, it is enough to use the following trivial lower bound
\begin{align*}
Z_n (m n, \rho n) \geq \frac{1}{2} \frac{n!}{(n-1)!} \frac{\left( \frac{\rho n + m n}{2}\right)^{n - 2}}{(n - 2)!}, 
\end{align*}
from which we obtain
\begin{align*}
\frac{1}{n} \ln Z_n ( m n, \rho n) \geq  \frac{1}{n} \ln \frac{1}{2} + \frac{n - 2}{n} \ln \frac{\rho + m}{2} + \frac{1}{n} \ln \frac{n^{n-1}}{(n-2)!} .
\end{align*}
It follows that
\begin{align*}
\liminf_{n \to \infty} \frac{1}{n} \ln Z_n (m n, \rho n) \geq \ln \frac{\rho + m}{2} + 1 ,
\end{align*}
as desired.
\end{proof}
\noindent
We continue by consider the properties of the limiting entropy $f(\beta, \mu)$.
\begin{proof}[Proof of \cref{thm:limiting microcanonical entropy}] First, we observe that
\begin{align*}
f(\beta, \mu) = \ln \int_{-\infty}^\infty d\phi \ e^{- \beta \phi - \mu |\phi|} .
\end{align*}
From this form, it is apparent that $f$ is strictly convex on $\mathcal{A}$ and thus is a proper convex function on $\mathbb{R}^2$. For lower semi-continuity, if $(\beta,\mu) \in \mathbb{R}^2 \setminus \overline{\mathcal{A}}$, then $f$ is lower semi-continuous for trivial reasons, in addition, since $f$ is continuous on $\mathcal{A}$, it is also necessarily lower semi-continuous there. For the points in $(\beta, \mu) \in \overline{\partial \mathcal{A}}$, it is clear that these points are of the form $(\pm \mu', \mu')$ for $\mu' \geq 0$. It is easy to check that $\lim_{(\beta, \mu) \to (\pm \mu', \mu)} f(\beta, \mu) = \infty$, since $f(\beta, \mu)$ is either equal to infinity, or it is increasing without bound for points inside $\mathcal{A}$ approaching $(\pm \mu', \mu')$. 
\\
\\
As for the other properties, the non-empty interior of the domain of finiteness of $f$ is given by $\mathcal{A}$. The mapping $f$ is differentiable in $\mathcal{A}$. For steepness, which is the third property of being essentially smooth, observe that
\begin{align*}
|| \nabla [f] (\beta, \mu)|| = \frac{1}{\frac{1}{\mu + \beta} + \frac{1}{\mu - \beta}} \sqrt{\frac{2}{(\mu + \beta)^4} + \frac{2}{(\mu - \beta)^4}} .
\end{align*}
Since all norms on $\mathbb{R}^2$ are equivalent, it follows that there exists a constant $C > 0$ such that
\begin{align*}
\left( \frac{1}{(\mu + \beta)^4} + \frac{1}{(\mu - \beta)^4} \right)^\frac{1}{4} \geq C \left( \frac{1}{\mu + \beta} + \frac{1}{\mu - \beta} \right) .
\end{align*}
Using this estimate, it follows that
\begin{align*}
|| \nabla [f] (\beta, \mu)|| \geq \sqrt{2} C \left( \frac{1}{\mu + \beta} + \frac{1}{\mu - \beta} \right) .
\end{align*}
From this estimate it is now clear that if $(\beta, \mu) \to (\pm \mu', \mu')$ for $\mu' \geq 0$ for points inside $\mathcal{A}$, then clearly $\lim_{(\beta, \mu) \to (\pm \mu', \mu')} || \nabla [f] (\beta, \mu)|| = \infty$, which shows steepness.
\\
\\
In summary, we find that $f$ is a proper convex lower semi-continuous function of Legendre type.
\\
\\
For the next few computational steps, it is useful to introduce the change of variables $g : \mathbb{R}^2 \to \mathbb{R}^2$ given by $(\beta, \mu) \mapsto g (\beta, \mu) = (\mu + \beta, \mu - \beta)$ so that for $(\beta, \mu) \in \mathcal{A}$, we have
\begin{align*}
f (\beta, \mu) = \ln \left( \frac{1}{g_1 (\beta, \mu)}  + \frac{1}{g_2 (\beta, \mu)}\right) .
\end{align*}
We can now equivalently consider the function $f' : (0, \infty)^2 \to \mathbb{R}$ given by
\begin{align*}
f' (g_1, g_2) = \ln \left( \frac{1}{g_1} + \frac{1}{g_2} \right) ,
\end{align*}
so that $f \circ g^{-1} = f'$. For the function $f'$ it is easy to verify that
\begin{align*}
- \nabla [f'] (g_1, g_2) = \left( \frac{g_2}{g_1 (g_2 + g_1)}, \ \frac{g_1}{g_2 (g_2 + g_1)} \right) ,
\end{align*}
and the inverse map can be computed from
\begin{align*}
&(0, \infty)^2 \ni (a,b) = - \nabla [f'] (g_1, g_2) \\ &\iff (g_1, g_2) = \left( \frac{1}{\sqrt{a} (\sqrt{a} + \sqrt{b})}, \ \frac{1}{\sqrt{b} (\sqrt{a} + \sqrt{b})} \right) = (- \nabla [f'])^{-1} (a,b)  .
\end{align*}
This shows that $(- \nabla [f']) (0, \infty)^2 = (0, \infty)^2$. Finally, for $(a,b) \in (0, \infty)^2$, one can observe that
\begin{align*}
\inf_{(g_1,g_2) \in (0,\infty)^2} \{ a g_1 + b g_2  + f' (g_1, g_2) \} &= (- \nabla [f'])^{-1}_1 (a,b) a + (- \nabla [f'])^{-1}_2 (a,b) b + (f \circ (- \nabla [f']) \left( a, b \right) \\
&= 1 + \ln \left( (\sqrt{a} + \sqrt{b})^2 \right) .
\end{align*}
To return to the function $f$, we have
\begin{align*}
(- \nabla [f]) \mathcal{A} = (D[g])^T ((- \nabla [f']) g (\mathcal{A})) =  (D[g])^T ((- \nabla [f']) (0, \infty)^2)  = (D[g])^T (0, \infty)^2 = \mathcal{A} ,
\end{align*}
where $D[g]$ is the derivative of the map $g$. We can also compute the following
\begin{align*}
\inf_{(\beta, \mu) \in \mathbb{R}^2} \{ \beta m + \mu \rho + f (\beta, \mu) \} &= \inf_{(\beta, \mu) \in \mathcal{A}} \{ \beta m + \mu \rho + f (\beta, \mu) \} \\
&= \inf_{(g_1, g_2) \in g (\mathcal{A}) = (0, \infty)^2} \{  g^{-1}_1 (g_1, g_2) m + g^{-1}_2 (g_1, g_2) \rho +( f \circ g^{-1}) (g_1, g_2) \} \\
&= \inf_{(g_1, g_2) \in (0, \infty)^2} \left\{  \frac{g_1 - g_2}{2} m + \frac{g_1 + g_2}{2} \rho +f' (g_1, g_2) \right\}  \\
&= \inf_{(g_1, g_2) \in (0, \infty)^2} \left\{ \frac{\rho + m}{2} g_1 + \frac{\rho - m}{2} g_2 + f' (g_1, g_2) \right\}  \\
&= 1 + \ln \left( \left( \sqrt{\frac{\rho + m}{2}} + \sqrt{\frac{\rho - m}{2}}\right)^2 \right) .
\end{align*}
To finish, note that we can simply compute the gradient
\begin{align*}
- \nabla [f] (\beta, \mu) = \left( - \frac{2 \beta}{\mu^2 - \beta^2}, \ \frac{\mu^2 + \beta^2}{\mu (\mu^2 - \beta^2)}\right) ,
\end{align*}
but its inverse map is simpler to solve from the composite function $f'$. Doing so, we obtain
\begin{align*}
\mathcal{A} \ni (m, \rho) = - \nabla [f] (\beta, \mu) \iff (\beta, \mu) = \left(  - \frac{\rho}{m} \frac{1}{\sqrt{\rho^2 - m^2}} + \frac{1}{m}, \ \frac{1}{\sqrt{\rho^2 - m^2}}\right) = (- \nabla [f])^{-1} (m, \rho).
\end{align*} 
Compiling together all of these results, we find that $f$ is a proper convex lower semi-continuous function of Legendre type which satisfies $ (- \nabla [f]) \mathcal{A} = \mathcal{A}$, and, for $(m, \rho) \in \mathcal{A}$, we have
\begin{align*}
\inf_{(\beta, \mu) \in \mathbb{R}^2} \{ \beta m + \mu \rho + f (\beta, \mu) \} &= \inf_{(\beta, \mu) \in \mathcal{A}} \{ \beta m + \mu \rho + f (\beta, \mu) \} \\
&= \beta (m, \rho) m + \mu (m, \rho) \rho + f (\beta (m, \rho), \mu (m, \rho)) \\
&= 1 + \ln \left( \left( \sqrt{\frac{\rho + m}{2}} + \sqrt{\frac{\rho - m}{2}}\right)^2 \right) ,
\end{align*}
where
\begin{align*}
(\beta (m, \rho), \mu (m, \rho)) = (- \nabla [f])^{-1} (m, \rho) = \left(  - \frac{\rho}{m} \frac{1}{\sqrt{\rho^2 - m^2}} + \frac{1}{m}, \ \frac{1}{\sqrt{\rho^2 - m^2}}\right) .
\end{align*}
\end{proof}
\noindent
We begin with the proof of the half-constrained ensemble limiting entropy.
\begin{proof}[Proof of \cref{thm:half-constrained ensemble free energy}] Fix $\beta \in \mathbb{R}$, and consider the mapping $Q_n (g^\beta, \cdot) : (0, \infty) \to \mathbb{R}$ given by
\begin{align*}
Q_n (g^{\beta}, \rho) := \int_{\mathbb{R}^n} d \phi \ e^{- \beta M_n (\phi)} \delta (N_n (\phi) - \rho n) ,
\end{align*}
which, like \cref{def:partition function rigorous}, is to be understood as 
\begin{align*}
Q_n (g^\beta, \rho) = e^{- \beta \rho n} Z_n (\rho n, \rho n)  + e^{\beta \rho n} Z_n (- \rho n, \rho n) + \int_{-\rho}^\rho dm \ n e^{- \beta m n} Z_n (m n, \rho n) .
\end{align*}
By direct computation, using \cref{def:microcanonical partition function}, it follows that
\begin{align*}
\lim_{n \to \infty} \frac{1}{n} \ln \left(  e^{- \beta \rho n} Z_n (\rho n, \rho n) \right) = - \beta \rho + \ln \rho + 1, \ \lim_{n \to \infty} \frac{1}{n} \ln \left(  e^{ \beta \rho n} Z_n (- \rho n, \rho n) \right) = \beta \rho + \ln \rho + 1 .
\end{align*}
As for the mapping
\begin{align*}
\rho \mapsto \int_{-\rho}^\rho dm \ n e^{- \beta m n} Z_n (m n, \rho n) = \int_{\mathbb{R}} dm \ n \mathbbm{1}(|m| < \rho) e^{- \beta m n} Z_n (m n, \rho n) ,
\end{align*}
it is enough to notice that the individual mappings in the integrand
\begin{align*}
\mathbb{R}^2 \ni (m, \rho) \mapsto \left( \mathbbm{1}( |m| < \rho), e^{- \beta m n}, \ Z_n (m n, \rho n) \right) 
\end{align*}
are log-concave functions. To be more precise, the indicator function is the indicator of a convex set and is thus log-concave, the exponential function is trivially log-concave by direct computation, and, finally, the microcanonical partition function, which is to be understood as the microcanonical partition function on $\mathcal{A}$ extended beyond this set by setting its value to $0$, is log-concave by  \cref{thm:microcanonical entropy regularity}. It follows that that the mapping 
\begin{align*}
\rho \mapsto \int_{-\rho}^\rho dm \ n e^{- \beta m n} Z_n (m n, \rho n)
\end{align*}
is log-concave by the Prékopa–Leindler inequality or Prékopa's theorem, see \cite[Section 9]{Gardner2002}, since it is the marginal of a log-concave function. 
\\
\\
For pointwise uniform boundedness, we begin by observing that
\begin{align*}
 e^{- |\beta| \rho n} \int_{- \rho}^\rho dm \ n Z_n (m n, \rho n) \leq \int_{-\rho}^\rho dm \ n e^{- \beta m n} Z_n (m n, \rho n)  \leq e^{|\beta| \rho n} \int_{- \rho}^\rho dm \ n Z_n (m n, \rho n)
\end{align*}
and
\begin{align*}
\int_{- \rho}^\rho dm \ n Z_n (m n, \rho n) = \rho^{n-1} n^{n - 1}  \int_{- 1}^1 dm \ Z_n (m, 1) .
\end{align*}
We will use the beta function $B(z_1,z_2)$ given by
\begin{align*}
B(z_1,z_2) := \int_{0}^1 dt \ t^{z_1 - 1} (1 - t)^{z_2 - 1}
\end{align*}
for $\operatorname{Re}(z_1), \operatorname{Re}(z_2)  > 0$. By a change of variables, one can see that
\begin{align*}
B(z_1,z_2) = \frac{1}{2} \int_{- 1}^{1} dt \ \left( \frac{1 + t}{2} \right)^{z_1 - 1} \left( \frac{1  - t}{2} \right)^{z_2 - 1} .
\end{align*}
For integer values, we have the following identity
\begin{align*}
B(m,n) = \frac{(m-1)! (n-1)!}{(m + n - 1)!}
\end{align*}
from which it follows that
\begin{align*}
\int_{-1}^1 dm \ Z_n (m,1) = \sum_{k=1}^{n - 1} {n \choose k} \frac{B(k, n - k)}{ (k-1)! (n - k - 1)!} = \frac{1}{(n-1)!} \sum_{k=1}^{n-1} {n \choose k} = \frac{2^{n} - 2}{(n-1)!} .
\end{align*}
In summary, we have
\begin{align*}
e^{- |\beta| \rho n} \rho^{n-1} n^{n - 1} \frac{2^n - 2}{(n-1)!} \leq \int_{-\rho}^\rho dm \ n e^{- \beta m n} Z_n (m n, \rho n) \leq  e^{|\beta| \rho n} \rho^{n-1} n^{n - 1} \frac{2^n - 2}{(n-1)!} .
\end{align*}
Computing the limits, it follows that
\begin{align*}
- \infty < \liminf_{n \to \infty} \frac{1}{n} \ln \int_{-\rho}^\rho dm \ n e^{- \beta m n} Z_n (m n, \rho n) \leq \limsup_{n \to \infty} \frac{1}{n} \ln \int_{-\rho}^\rho dm \ n e^{- \beta m n} Z_n (m n, \rho n) < \infty ,
\end{align*}
from which the uniform pointwise boundedness follows.
\\
\\
For $\mu > |\beta|$, we can directly compute that
\begin{align*}
\int_0^\infty d \rho \ n e^{- \mu \rho n} \int_{-\rho}^\rho dm \ n e^{- \beta m n} Z_n (m n, \rho n) = \left( \frac{1}{\mu + \beta} + \frac{1}{\mu - \beta} \right)^n - \left( \frac{1}{\mu + \beta} \right)^n-  \left( \frac{1}{\mu - \beta} \right)^n .
\end{align*} 
For any other value of $\mu$, it is clear that the above integral is infinite. It follows that the limit and subsequent mapping given by
\begin{align*}
\mu \mapsto \lim_{n \to \infty} \frac{1}{n} \ln \int_0^\infty d \rho n e^{- \mu \rho n} \int_{-\rho}^\rho dm \ n e^{- \beta m n} Z_n (m n, \rho n) = f(\beta, \mu),
\end{align*}
exists and has a domain of finiteness given by the half-infinite interval $(|\beta|, \infty)$. By using the properties of the full map $(\beta, \mu) \mapsto f (\beta, \mu)$, already verified and computed in \cref{thm:free energy properties}, one can verify that the mapping $\mu \mapsto f (\beta, \mu)$ for fixed $\beta$ is a proper convex lower semi-continuous function of Legendre type that satisfies $- D[f(\beta, \cdot)] = (0, \infty)$. By \cref{thm:entropy convergence}, for any $\rho > 0$, it follows that
\begin{align*}
\lim_{n \to \infty} \frac{1}{n} \ln \int_{-\rho}^\rho dm \ n e^{- \beta m n} Z_n (m n, \rho n) = \inf_{\mu > |\beta|} \{ \mu \rho + f (\beta, \mu) \} .
\end{align*}
To continue, by \cref{thm:free energy properties}, we have
\begin{align*}
f(\beta, \mu) = \inf_{(m, \rho) \in \mathcal{A}} \{ \beta m + \mu \rho - s (m, \rho) \} = \inf_{\rho > 0} \left\{ \mu \rho + \inf_{|m| < \rho} \{ \beta m - s (m, \rho) \}\right\} ,
\end{align*}
so that
\begin{align*}
\inf_{\mu > |\beta|} \{ \mu \rho + f (\beta, \mu) \} = -\inf_{|m| < \rho} \{ \beta m - s (m, \rho) \} = \sup_{|m| < 1} \{ s (m, \rho) - \beta m \} .
\end{align*}
For the rate function, the scaled logarithmic moment generating function $\Lambda : \mathbb{R} \to [- \infty, \infty]$ of a sequence of random variables with distributions given by $\left\{ \kappa_n^\beta \right\}_{n \in \mathbb{N}}$ is given by
\begin{align*}
\Lambda(t) := \lim_{n \to \infty} \frac{1}{n} \ln \frac{Q_n (\beta - t)}{Q_n (\beta)} &= \sup_{|m| < 1} \{ s (m, 1) - (\beta - t) m \} - \sup_{|m| < 1} \{ s (m, 1) - \beta m \} \\
&= \sup_{|m| < 1} \{ tm - (-(s (m, 1) - \beta m)) \} - \sup_{|m| < 1} \{ s (m, 1) - \beta m \} .
\end{align*}
We can identify the first term on the last line as the convex conjugate of the restriction of a proper convex lower semi-continuous function of Legendre type with an interior of the domain of finiteness given by $(-1,1)$. From the form of the function $s(m,1)$, for $m \in (-1,1)$, we immediately see that
\begin{align*}
\lim_{m \to {\pm 1}^\mp} s (m,1) = 1 .
\end{align*}
Defining $s (\pm 1, 1) = 1$ yields a continuous extension of $s(m,1)$ from $(-1,1)$ to $[-1,1]$, and we will consider it so from now on. The extended mapping given by
\begin{align*}
\mathbb{R} \ni m \mapsto \begin{cases} s (m,1), &\ m \in [-1,1] , \\
-\infty, &\ m \not \in [-1,1] ,    \end{cases}
\end{align*}
is upper semi-continuous, and we will consider this the redefinition of $s (m,1)$ to be understood now as not necessarily finite function on $\mathbb{R}$. Compiling all of this together, it follows that the mapping $\mathbb{R} \ni m \mapsto - (s (m, 1) - \beta m)$ defines a proper convex lower semi-continuous function of Legendre type, and thus the convex conjugate is involutive from which it follows that
\begin{align*}
\Lambda^* (m) = \sup_{|m| < 1} \{ s (m,1) - \beta m\} - (s (m,1) - \beta m) ,
\end{align*}
which is the rate function of $\{ \kappa_n^\beta \}_{n \in \mathbb{N}}$.
\end{proof}
\noindent
We prove the proof of the limit point result.
\begin{proof}[Proof of \cref{thm:partial limiting point}]
Using \cref{thm:limiting points large deviations}, let $\{ \kappa^g_{n_k}\}_{k \in \mathbb{N}}$ be a weakly convergent subsequence with a limit $\kappa$. Since $M^* (\psi^g)$ is a compact subset of $(-1,1)$, it follows that there exists $a := \min M^* (\psi^g)$ and $b := \max M^* (\psi^g)$. There exists $\delta > 0$ such that $\operatorname{supp} (\kappa) \subset M^* (\psi^g) \subset [a - \delta, b + \delta] \subset (-1,1)$. Since $\operatorname{supp} (\kappa) \subset [a - \delta, b + \delta]$, we deduce that $\kappa ([a - \delta, b + \delta]) = 1$, and, since $\partial [a - \delta, b + \delta] \cap \operatorname{supp} (\mu) \subset \{ a - \delta, b + \delta\} \cap M^* (\psi^g) = \emptyset$, we see that $\kappa (\partial [a - \delta, b + \delta]) = 0$. It follows that $[a - \delta, b + \delta] \subset (-1,1)$ is a continuity set of $\kappa$, and we can apply \cref{thm:general convergence} along this subsequence with \cref{thm:locally uniform convergence expectations} to obtain the result. 
\end{proof}
\subsection{Asymptotics of the weights} \label{sec:Asymptotics of the weights}
\noindent 
We first establish the Laplace-type representation of the microcanonical partition function.
\begin{proof}[Proof of \cref{thm:microcanonical partition function integral}]
The microcanonical partition function can be written as
\begin{align*}
Z_n (m n, \rho n) = \frac{2 n^{n - 2} n!}{(\rho^2 - m^2) n^2} \sum_{k=1}^{n - 1} \frac{\left(\frac{\rho + m}{2} \right)^k}{(k-1)! k!} \frac{\left(\frac{\rho - m}{2} \right)^{n-k}}{(n - k - 1)! (n - k)!} ,
\end{align*}
which one can recognize as the convolution of two sequences with some factors in front. We consider the generating function $G : \mathbb{C} \to \mathbb{C}$ given by
\begin{align*}
G(z) &:= \sum_{n=2}^\infty \frac{(\rho^2 - m^2) n^2}{2 n^{n - 2} n!} Z_n ( m n, \rho n) \left( \frac{z^2}{4} \right)^{n} \\ &= \left( \sum_{n=2}^\infty \frac{\left( \frac{1}{4} \left( \sqrt{\frac{\rho + m}{2}} z \right)^2 \right)^{n}}{n!(n-1)!}\right) \left( \sum_{n=2}^\infty \frac{\left( \frac{1}{4} \left( \sqrt{\frac{\rho - m}{2}} z \right)^2 \right)^{n}}{n!(n-1)!}\right) .
\end{align*} 
One can verify that the convolution yields a Cauchy product, and that the power series on the right define entire functions with absolutely convergent power series. We have the standard relation between the derivatives of $G$ and its power series coefficients
\begin{align*}
\frac{G^{(2n)} (0)}{(2n)!} = \frac{(\rho^2 - m^2) n^2}{2 n^{n - 2} n! 4^{n}} Z_n ( m n, \rho n) \iff Z_n ( m n, \rho n) = \frac{2^{2n+1} n^{n - 2} n!}{(\rho^2 - m^2) n^2} \frac{G^{(2n)} (0)}{(2n)!} .
\end{align*} 
Next, using the modified Bessel function of the first kind $I_\nu (z)$ given by
\begin{align*}
I_\nu (z) := \left( \frac{1}{2} z \right)^\nu \sum_{n=0}^\infty \frac{\left( \frac{z^2}{4} \right)^n}{n! \Gamma (\nu + n + 1)} ,
\end{align*}
where $\nu \in \mathbb{Z}$, and we have
\begin{align*}
G(z) = \frac{1}{4} \sqrt{\frac{\rho^2 - m^2}{4}} z^2 I_{-1} \left( \sqrt{\frac{\rho + m}{2}} z\right) I_{-1} \left( \sqrt{\frac{\rho - m}{2}} z\right) .
\end{align*}
Using the integral representation, see \cite[Chapter 9]{Abramowitz1974}, given by
\begin{align*}
I_{\nu} (z) := \frac{1}{\pi} \int_0^\pi d \theta \ \cos(\nu \theta) e^{z \cos \theta}, 
\end{align*}
we see that
\begin{align*}
G(z) = \frac{1}{4} \sqrt{\frac{\rho^2 - m^2}{4}} z^2 \frac{1}{\pi^2} \int_0^\pi d \theta_1 \int_0^\pi d \theta_2 \ \cos \theta_1 \cos \theta_2 e^{z \left( \sqrt{\frac{\rho + m}{2}} \cos \theta_1 + \sqrt{\frac{\rho - m}{2}} \cos \theta_2 \right)} .
\end{align*}
Taking derivatives, using the general Leibniz rule, we obtain
\begin{align*}
&G^{(2n)} (0) \\ &= \frac{1}{2} \sqrt{\frac{\rho^2 - m^2}{4}}  {2 n \choose 2} \frac{1}{\pi^2} \int_0^\pi d \theta_1 \int_0^\pi d \theta_2 \ \\ &\times \cos \theta_1 \cos \theta_2 \left(  \sqrt{\frac{\rho + m}{2}} \cos \theta_1 + \sqrt{\frac{\rho - m}{2}} \cos \theta_2 \right)^{2 n - 2} ,
\end{align*}
from which it follows that
\begin{align*}
&Z_n (m n, \rho n)  \\ &= \frac{2^{2n - 1} n^{n - 2} n!}{(2n)! \sqrt{\rho^2 - m^2} n^2} {2n \choose 2}  \frac{1}{\pi^2} \int_0^\pi d \theta_1 \int_0^\pi d \theta_2 \ \\ &\times \cos \theta_1 \cos \theta_2 \left(  \sqrt{\frac{\rho + m}{2}} \cos \theta_1 + \sqrt{\frac{\rho - m}{2}} \cos \theta_2 \right)^{2 n - 2}  .
\end{align*}
By using the given from of the overloaded $s$ function and simplifying, we obtain the desired representation.
\end{proof}
\noindent
We present the proof of the local asymptotics of the overloaded ${\psi^g}$ function.
\begin{proof}
By computing the critical points of the overloaded ${\psi^g}$ function, we see that there is precisely one critical point in the given set in the assumptions, and it is given by $(m^*,0,0)$. For this particular critical point, it is easy to see that any odd partial derivative with respect to either $\theta_1$ or $\theta_2$ is vanishing.
\\
\\
By developing ${\psi^g}$ to second order in $(\theta_1, \theta_2)$, and $(2k)$:th order in $m$, it follows that
\begin{align*}
{\psi^g} (m^* + m, \theta_1, \theta_2) &= {\psi^g} (m^*) + \frac{1}{2} \partial_2^2 [{\psi^g}] (m^*, 0, 0) \theta_1^2 + \frac{1}{2} \partial_{3}^2 [{\psi^g}] (m^*, 0, 0) \theta_2^2 + \frac{1}{(2k)!} \partial^{2k}[\psi^g] (m^*) m^{2k}\\
&+ \sum_{|\alpha| = 3, \ \alpha_1 \not \in \{ 2, 3 \}} R_\alpha (m, \theta_1, \theta_2) (m, \theta_1, \theta_2)^\alpha + R_{(2k + 1,0,0)} (m, \theta_1, \theta_2) m^{2k + 1} ,
\end{align*}
where
\begin{align*}
R_\alpha (m, \theta_1, \theta_2) = \frac{|\alpha|}{\alpha!} \int_0^1 dt \ (1 - t)^{|\alpha| - 1} \partial_\alpha [{\psi^g}] ((m^*, 0,0) + t (m, \theta_1, \theta_2)) .
\end{align*}
\end{proof}
\noindent
We can now prove the full Laplace method for the mixture measures.
\begin{proof} Let us first remark that in the following proof, we will frequently use the statement for small enough $\delta > 0$ something holds. In the context of this proof, we repeat this to imply that there is a series of finite choice of $\delta > 0$ small enough such that all the conditions required will hold. In reality this proof should be worked through "backwards" so that the choice of $\delta > 0$ is clear.
\\
\\
We begin by noting that
\begin{align*}
&\frac{(2n)! n^2 \pi^2}{2^{2n - 1} n^{n - 2} n! {2n \choose 2} e^{-(n-1)}} \int_{m^* - \delta}^{m + \delta} dm \ e^{n (g(m) + s_n (m,1))} \\ &= \int_{m^* - \delta}^{m + \delta} dm \int_{0}^\pi d \theta_1 \int_0^\pi d \theta_2 \ \frac{\cos \theta_1 \cos \theta_2 e^{g(m)}}{\sqrt{1 - m^2}} e^{(n-1) ({\psi^g}(m, \theta_1, \theta_2))}
\end{align*}
and by using the symmetries of the trigonometric functions, it follows that
\begin{align*}
 &\int_{m^* - \delta}^{m^* + \delta} dm \int_0^\pi d \theta_1 \int_0^\pi d \theta_2 \frac{e^{g (m)} \cos \theta_1 \cos \theta_2}{\sqrt{1 - m^2}} e^{(n-1) {\psi^g} (m, \theta_1, \theta_2)} \\
&= \int_{m^* - \delta}^{m^* + \delta} dm \int_{- \frac{\pi}{2}}^{\frac{\pi}{2}} d \theta_1 \int_{- \frac{\pi}{2}}^{\frac{\pi}{2}} d \theta_2 \frac{e^{g (m)} \cos \theta_1 \cos \theta_2}{\sqrt{1 - m^2}} e^{(n-1) {\psi^g} (m, \theta_1, \theta_2)} \\
&- \int_{m^* - \delta}^{m^* + \delta} dm \int_0^{\frac{\pi}{2}} d \theta_1 \int_0^{\frac{\pi}{2}} d \theta_2 \frac{e^{g (m)} \sin \theta_1 \cos \theta_2}{\sqrt{1 - m^2}} e^{(n-1) {\psi^g} (m, \theta_1 + \frac{\pi}{2}, \theta_2)}  \\ &-\int_{m^* - \delta}^{m^* + \delta} dm \int_0^{\frac{\pi}{2}} d \theta_1 \int_0^{\frac{\pi}{2}} d \theta_2 \frac{e^{g (m)} \cos \theta_1 \sin \theta_2}{\sqrt{1 - m^2}} e^{(n-1) {\psi^g} (m, \theta_1 , \theta_2 + \frac{\pi}{2})} .
\end{align*}
We want to show that the first integral on the second line of this manipulation is exponentially dominant. To save space, denote the integrals as follows
\begin{align*}
I_1 (n) :=  \int_{m^* - \delta}^{m^* + \delta} dm \int_{- \frac{\pi}{2}}^{\frac{\pi}{2}} d \theta_1 \int_{- \frac{\pi}{2}}^{\frac{\pi}{2}} d \theta_2 \frac{e^{g (m)} \cos \theta_1 \cos \theta_2}{\sqrt{1 - m^2}} e^{(n-1) {\psi^g} (m, \theta_1, \theta_2)} \ , \\
I_2 (n) := \int_{m^* - \delta}^{m^* + \delta} dm \int_0^{\frac{\pi}{2}} d \theta_1 \int_0^{\frac{\pi}{2}} d \theta_2 \frac{e^{g (m)} \sin \theta_1 \cos \theta_2}{\sqrt{1 - m^2}} e^{(n-1) {\psi^g} (m, \theta_1 + \frac{\pi}{2}, \theta_2)} \ , \\
I_3 (n) := \int_{m^* - \delta}^{m^* + \delta} dm \int_0^{\frac{\pi}{2}} d \theta_1 \int_0^{\frac{\pi}{2}} d \theta_2 \frac{e^{g (m)} \cos \theta_1 \sin \theta_2}{\sqrt{1 - m^2}} e^{(n-1) {\psi^g} (m, \theta_1 , \theta_2 + \frac{\pi}{2})}  \ .
\end{align*} 
For the terms $I_2$ and $I_3$, observe that
\begin{align*}
\left| \sqrt{\frac{1 + m}{2}} \sin \alpha - \sqrt{\frac{1 - m}{2}} \cos \beta \right| \leq \max \left\{  \sqrt{\frac{1 + m}{2}}, \sqrt{\frac{1 - m}{2}} \right\} < \sqrt{\frac{1 + m}{2}} +  \sqrt{\frac{1 - m}{2}}
\end{align*} 
for any $\alpha, \beta \in [0, \frac{\pi}{2}]$ and $m \in (m^* - \delta, m^* + \delta)$. Using this property, one can check that 
\begin{align*}
M_2 (\delta) &:= \max_{(m, \theta_1, \theta_2) \in (m^* - \delta, m^* + \delta) \times [0, \frac{\pi}{2}] \times [0, \frac{\pi}{2}]} {\psi^g} \left(m, \theta_1 + \frac{\pi}{2}, \theta_2 \right) \\ &\leq \max_{m \in (m^* - \delta, m^* + \delta)} \left\{ g(m) + 1 + \ln \left( \left(  \max \left\{  \sqrt{\frac{1 + m}{2}}, \sqrt{\frac{1 - m}{2}} \right\}\right)^2 \right) \right\} .
\end{align*}
By continuity of the function inside the maximum, one can check that
\begin{align*}
\lim_{\delta' \to 0^+} M_2 (\delta') < M_1 (\delta) :=  \max_{(m, \theta_1, \theta_2) \in (m^* - \delta, m^* + \delta) \times [- \frac{\pi}{2}, \frac{\pi}{2}] \times [- \frac{\pi}{2}, \frac{\pi}{2}]} {\psi^g} \left(m, \theta_1, \theta_2 \right) = {\psi^g}(m^*) ,
\end{align*}
from which it follows that for small enough $\delta > 0$, we have $M_2 (\delta) < M_1 (\delta)$. One can verify in the same way that
\begin{align*}
M_3 (\delta) := \max_{(m, \theta_1, \theta_2) \in (m^* - \delta, m^* + \delta) \times [0, \frac{\pi}{2}] \times [0, \frac{\pi}{2}]} {\psi^g} \left(m, \theta_1 , \theta_2 + \frac{\pi}{2} \right) < M_1 (\delta)
\end{align*}
for small enough $\delta > 0$. For such $\delta$, it follows that
\begin{align*} 
\lim_{n \to \infty} \frac{1}{n} \ln I_{2/3} (n) = M_{2/3} (\delta) < M_1 (\delta) = \lim_{n \to \infty} \frac{1}{n} \ln I_{1} (n),
\end{align*}
which shows that $I_1 (n)$ exponentially dominates $I_{2/3} (n)$.
\\
\\
To continue, we have
\begin{align*}
\frac{n^{\frac{1}{2k} + 1} (I_1 (n) - I_2 (n) - I_3 (n))}{e^{n M_1}} = \frac{n^{\frac{1}{2k} + 1} I_1 (n)}{e^{n M_1}} - n^{\frac{1}{2k} + 1} e^{n \left( \frac{1}{n} \ln I_{2} (n) - M_1 \right)} - n^{\frac{1}{2k} + 1} e^{n \left( \frac{1}{n} \ln I_{3} (n) - M_1 \right)} .
\end{align*}
It is now clear that in the limit the terms on the right of the $I_1(n)$ term vanish since they are exponentially small. As for the limit of the integral $I_1 (n)$, it is solved by a routine application of Laplace's method using the asymptotics developed in \cref{thm:tilting function taylor polynomial}. First, however, we must split the integral $I_1 (n)$ with respect to the angular variables. Denote
\begin{align*}
f(m, \theta_1, \theta_2) := \frac{e^{g(m)} \cos \theta_1 \cos \theta_2}{\sqrt{1 - m^2}} .
\end{align*} 
Since ${\psi^g}$ attains it unique maximum at $(m^*, 0, 0)$, it follows that
\begin{align*}
&\lim_{n \to \infty} \frac{1}{n} \ln \int_{m^* - \delta}^{ m^* + \delta} dm \ \int_{\left([- \delta, \delta] \times [- \delta, \delta] \right)^c} d \theta_1 d \theta_2 \ f (m, \theta_1, \theta_2) e^{(n-1) {\psi^g} (m, \theta_1, \theta_2)} \\ &= \sup_{m \in [m^* - \delta, m^* + \delta] \times \left([- \delta, \delta] \times [- \delta, \delta] \right)^c} {\psi^g} (m, \theta_1, \theta_2) < M_1 .
\end{align*} 
If we denote
\begin{align*}
I_{1, \delta} (n) := \int_{m^* - \delta}^{ m^* + \delta} dm \ \int_{-\delta}^\delta d \theta_1 \int_{- \delta}^\delta d \theta_2 \ f (m, \theta_1, \theta_2) e^{(n-1) {\psi^g} (m, \theta_1, \theta_2)} ,
\end{align*}
we have
\begin{align*}
\frac{n^{\frac{1}{2k} + 1} I_1 (n)}{e^{n M_1}} = \frac{n^{\frac{1}{2k} + 1} I_{1, \delta} (n)}{e^{n M_1}} + n^{\frac{1}{2k} + 1} e^{n \left( \frac{1}{n} \ln (I_1 (n) - I_{1, \delta} (n)) - M_1 \right)} .
\end{align*}
Again, since the right hand side contains exponentially decreasing terms, the asymptotics will be determined by the first term on the right. Finally, by changing variables, observe that
\begin{align*}
&\frac{n^{\frac{1}{2k} + 1} I_{1, \delta} (n)}{e^{(n-1) M_1}} \\ &= \int_{- \delta n^\frac{1}{2k}}^{\delta n^{\frac{1}{2 k}}} dm \int_{- \frac{\pi}{2} n^{\frac{1}{2}}}^{ \frac{\pi}{2} n^{\frac{1}{2}}} d \theta_1 \int_{- \frac{\pi}{2} n^{\frac{1}{2}}}^{ \frac{\pi}{2} n^{\frac{1}{2}}} d \theta_2 \ f \left(m^* + \frac{m}{n^{\frac{1}{2k}}}, \frac{\theta_1}{n^{\frac{1}{2}}}, \frac{\theta_2}{n^{\frac{1}{2}}} \right) e^{(n-1) \left( {\psi^g} \left(m^* + \frac{m}{n^{\frac{1}{2k}}}, \frac{\theta_1}{n^\frac{1}{2}}, \frac{\theta_2}{n^{\frac{1}{2}}}\right) - {\psi^g} (m^*)\right)} . 
\end{align*}
If one looks at the remainder term displayed in \cref{thm:tilting function taylor polynomial}, one finds that
\begin{align*}
&\left| \sum_{|\alpha| = 3, \ \alpha_1 \not \in \{ 2, 3 \}} R_\alpha (m, \theta_1, \theta_2) (m, \theta_1, \theta_2)^\alpha \right| \\ &\leq \max_{(m, \theta_1, \theta_2) \in [- \delta, \delta]^3, \ |\alpha| = 3, \ \alpha_1 \not \in \{ 2, 3 \}|} |R_\alpha (m, \theta_1, \theta_2)| \sum_{|\alpha| = 3, \ \alpha_1 \not \in \{ 2, 3\}} |(m, \theta_1, \theta_2)^\alpha| \\
&\leq  \max_{(m, \theta_1, \theta_2) \in [- \delta, \delta]^3, \ |\alpha| = 3, \ \alpha_1 \not \in \{ 2, 3 \}|} |R_\alpha (m, \theta_1, \theta_2)| (A |\theta_1|^3 + B \theta_1^2 |\theta_2| + C |\theta_1| \theta_2^2 + D |\theta_2|^3 + E |m| |\theta_1| |\theta_2|) \\
&\leq  \left( \delta F \max_{(m, \theta_1, \theta_2) \in [- \delta, \delta]^3, \ |\alpha| = 3, \ \alpha_1 \not \in \{ 2, 3 \}|} |R_\alpha (m, \theta_1, \theta_2)|  \right) (\theta_1^2 + \theta_2^2) ,
\end{align*}
and
\begin{align*}
\left| R_{(2k + 1,0,0)} (m, \theta_1, \theta_2) m^{2k + 1} \right| \leq  \left( \delta \max_{(m, \theta_1, \theta_2) \in [- \delta, \delta]^3} |R_{(2k+1,0,0)} (m, \theta_1, \theta_2)| \right) m^{2k} ,
\end{align*}
where $A,B,C,D,E,F > 0$ are all positive constants. For $\delta$ satisfying
\begin{align*}
\delta F \max_{(m, \theta_1, \theta_2) \in [- \delta, \delta]^3, \ |\alpha| = 3, \ \alpha_1 \not \in \{ 2, 3 \}|} |R_\alpha (m, \theta_1, \theta_2)| < \max \left \{- \frac{1}{2} \partial_2^2 [{\psi^g}] (m^*, 0, 0)m, - \frac{1}{2} \partial_{3}^2 [{\psi^g}] (m^*, 0, 0) \right\} ,
\end{align*}
and
\begin{align*}
\delta \max_{(m, \theta_1, \theta_2) \in [- \delta, \delta]^3} |R_{(2k+1,0,0)} (m, \theta_1, \theta_2)|\leq - \frac{1}{(2k)!} \partial^{2k} [\psi^g](m^*) .
\end{align*}
Ultimately, for $\delta > 0$ chosen small enough so as to satisfy the finite number of conditions given previously, using the error bounds above, by dominated convergence, it follows that
\begin{align*}
&\lim_{n \to \infty} \int_{- \delta n^\frac{1}{2k}}^{\delta n^{\frac{1}{2 k}}} dm \int_{- \frac{\pi}{2} n^{\frac{1}{2}}}^{ \frac{\pi}{2} n^{\frac{1}{2}}} d \theta_1 \int_{- \frac{\pi}{2} n^{\frac{1}{2}}}^{ \frac{\pi}{2} n^{\frac{1}{2}}} d \theta_2 \ f \left(m^* + \frac{m}{n^{\frac{1}{2k}}}, \frac{\theta_1}{n^{\frac{1}{2}}}, \frac{\theta_2}{n^{\frac{1}{2}}} \right) e^{(n-1) \left( {\psi^g} \left(m^* + \frac{m}{n^{\frac{1}{2k}}}, \frac{\theta_1}{n^\frac{1}{2}}, \frac{\theta_2}{n^{\frac{1}{2}}}\right) - {\psi^g} (m^*)\right)} \\
&= f(m^*, 0, 0) \int_{\mathbb{R}^3} d \theta_1 d \theta_2 d m \ e^{\frac{1}{2} \partial_2^2 [{\psi^g}] (m^*, 0, 0) \theta_1^2 + \frac{1}{2} \partial_{3}^2 [{\psi^g}] (m^*, 0, 0) \theta_2^2 + \frac{1}{(2k)!} \partial^{2k} [\psi^g] (m^*) m^{2k}} .
\end{align*}
Combining all of these results together, it follows that
\begin{align*} &\lim_{n \to \infty}\frac{n^{\frac{1}{2k} + 1}\int_{m^* - \delta}^{m + \delta} dm \ e^{n (g(m) + s_n (m,1))}}{e^{n {\psi^g}(m^*)}} \frac{(2n)! n^2 \pi^2}{2^{2n - 1} n^{n - 2} n! {2n \choose 2} e^{-(n-1)}}  \\ &= \frac{e^{g(m^*)}}{e^{{\psi^g} (m^*)} \sqrt{1 - {m^*}^2}} \int_{\mathbb{R}^3} d \theta_1 d \theta_2 d m \ e^{\frac{1}{2} \partial_2^2 [{\psi^g}] (m^*, 0, 0) \theta_1^2 + \frac{1}{2} \partial_{3}^2 [{\psi^g}] (m^*, 0, 0) \theta_2^2 + \frac{1}{(2k)!} \partial^{2k} [\psi^g] (m^*) m^{2k}} .
\end{align*}
\end{proof}


\begin{thebibliography}{10}

\bibitem{Abramowitz1974}
Milton Abramowitz.
\newblock {\em Handbook of Mathematical Functions, With Formulas, Graphs, and
  Mathematical Tables,}.
\newblock Dover Publications, Inc., USA, 1974.

\bibitem{Anari2021}
Nima Anari, Shayan~Oveis Gharan, and Cynthia Vinzant.
\newblock Log-concave polynomials, i: Entropy and a deterministic approximation
  algorithm for counting bases of matroids.
\newblock {\em Duke Mathematical Journal}, 170(16), nov 2021.

\bibitem{Berlin1952}
T.~H. Berlin and M.~Kac.
\newblock The spherical model of a ferromagnet.
\newblock {\em Physical Review}, 86(6):821--835, jun 1952.

\bibitem{Braenden2020}
Petter Bränd{\'{e}}n and June Huh.
\newblock Lorentzian polynomials.
\newblock {\em Annals of Mathematics}, 192(3), nov 2020.

\bibitem{Caputo2003}
Pietro Caputo.
\newblock Uniform poincar{\'{e}} inequalities for unbounded conservative spin
  systems: the non-interacting case.
\newblock {\em Stochastic Processes and their Applications}, 106(2):223--244,
  aug 2003.

\bibitem{Chatterjee2017}
Sourav Chatterjee.
\newblock A note about the uniform distribution on the intersection of a
  simplex and a sphere.
\newblock {\em Journal of Topology and Analysis}, 09(04):717--738, aug 2017.

\bibitem{Matos1991}
J.~M. G.~Amaro de~Matos and J.~Fernando Perez.
\newblock Fluctuations in the curie-weiss version of the random field ising
  model.
\newblock {\em Journal of Statistical Physics}, 62(3-4):587--608, feb 1991.

\bibitem{Hollander2008}
Frank den Hollander.
\newblock {\em Large Deviations}.
\newblock American Mathematical Society, jun 2008.

\bibitem{Eisele1988}
Theodor Eisele and Richard~S. Ellis.
\newblock Multiple phase transitions in the generalized curie-weiss model.
\newblock {\em Journal of Statistical Physics}, 52(1-2):161--202, jul 1988.

\bibitem{Ellis2006}
Richard~S. Ellis.
\newblock {\em Entropy, Large Deviations, and Statistical Mechanics}.
\newblock Springer Berlin Heidelberg, 2006.

\bibitem{Ellis1978}
Richard~S. Ellis and Charles~M. Newman.
\newblock The statistics of curie-weiss models.
\newblock {\em Journal of Statistical Physics}, 19(2):149--161, aug 1978.

\bibitem{Friedli2017}
Sacha Friedli and Yvan Velenik.
\newblock {\em Statistical Mechanics of Lattice Systems}.
\newblock Cambridge University Press, nov 2017.

\bibitem{Gardner2002}
R.~J. Gardner.
\newblock The brunn-minkowski inequality.
\newblock {\em Bulletin of the American Mathematical Society}, 39(03):355--406,
  apr 2002.

\bibitem{Georgii2011}
Hans-Otto Georgii.
\newblock {\em Gibbs Measures and Phase Transitions}.
\newblock {DE} {GRUYTER}, may 2011.

\bibitem{Grosskinsky2008}
Stefan Gro{\ss}kinsky.
\newblock Equivalence of ensembles for two-species zero-range invariant
  measures.
\newblock {\em Stochastic Processes and their Applications}, 118(8):1322--1350,
  aug 2008.

\bibitem{Kastner2006}
Michael Kastner and Oliver Schnetz.
\newblock On the mean-field spherical model.
\newblock {\em Journal of Statistical Physics}, 122(6):1195--1214, mar 2006.

\bibitem{Koskinen2023}
Kalle Koskinen.
\newblock Infinite volume gibbs states and metastates of the random field
  mean-field spherical model.
\newblock {\em Journal of Statistical Physics}, 190(3), feb 2023.

\bibitem{Koskinen2020}
Kalle Koskinen and Jani Lukkarinen.
\newblock Estimation of local microcanonical averages in two lattice mean-field
  models using coupling techniques.
\newblock {\em Journal of Statistical Physics}, 180(1-6):1206--1251, jul 2020.

\bibitem{Lukkarinen2019}
Jani Lukkarinen.
\newblock Multi-state condensation in berlin{\textendash}kac spherical models.
\newblock {\em Communications in Mathematical Physics}, 373(1):389--433, dec
  2019.

\bibitem{Nam2020}
Kyeongsik Nam.
\newblock Large deviations and localization of the microcanonical ensembles
  given by multiple constraints.
\newblock {\em The Annals of Probability}, 48(5), sep 2020.

\bibitem{Rockafellar1997}
Ralph~Tyrell Rockafellar.
\newblock {\em Convex Analysis}.
\newblock Princeton University Press, 1997.

\bibitem{Sason2016}
Igal Sason and Sergio Verdu.
\newblock {\textdollar}f{\textdollar} -divergence inequalities.
\newblock {\em {IEEE} Transactions on Information Theory}, 62(11):5973--6006,
  nov 2016.

\bibitem{Touchette2015}
Hugo Touchette.
\newblock Equivalence and nonequivalence of ensembles: Thermodynamic,
  macrostate, and measure levels.
\newblock {\em Journal of Statistical Physics}, 159(5):987--1016, feb 2015.

\bibitem{Wong2001}
R.~Wong.
\newblock {\em Asymptotic Approximations of Integrals}.
\newblock Society for Industrial and Applied Mathematics, jan 2001.

\end{thebibliography}
\end{document}